\documentclass[draft]{article}
\usepackage{amsmath,amsfonts,latexsym, amssymb, mathrsfs}

\usepackage{color}

\def\cH{{\mathcal H}}

\newcommand{\bm}{{\bf{m}}}
\newcommand{\cN}{{\mathcal N}}

\newcommand{\non}{\nonumber}

\newcommand{\fa}{{\frak  a}}
\newcommand{\fb}{{\frak  b}}
\newcommand{\fc}{{\frak c}}

\newcommand{\wt}{\widetilde}
\newcommand{\wh}{\widehat}

\newcommand{\e}{\varepsilon}

\newcommand{\pt}{\partial}
\newcommand{\rd}{{\rm d}}
\newcommand{\bR}{{\mathbb R}}

\newcommand{\ba}{{\bf{a}}}
\newcommand{\bb}{{\bf{b}}}
\newcommand{\bx}{{\bf{x}}}
\newcommand{\by}{{\bf{y}}}
\newcommand{\bu}{{\bf{u}}}
\newcommand{\bv}{{\bf{v}}}

\newcommand{\al}{\alpha}

\newcommand{\be}{\begin{equation}}
\newcommand{\ee}{\end{equation}}

\newcommand{\la}{\lambda}

\newcommand{\om}{{\omega}}

\newcommand{\cA}{{\cal A}}

\newcommand{\cG}{{\cal G}}

\newcommand{\im}{{\text Im }}
\newcommand{\E}{{\mathbb E }}
\newcommand{\R}{{\mathbb R }}
\newcommand{\N}{{\mathbb N}}
\renewcommand{\P}{{\mathbb P}}

\newtheorem{theorem}{Theorem}
\newtheorem{corollary}[theorem]{Corollary}
\newtheorem{lemma}[theorem]{Lemma}
\newtheorem{proposition}[theorem]{Proposition}

\newtheorem{remark}[theorem]{Remark}

\newcommand{\qed}{\hfill\fbox{}\par\vspace{0.3mm}}
\newenvironment{proof}{{\bf Proof.}} {\hfill\qed}

\numberwithin{equation}{section}
\numberwithin{theorem}{section}
\numberwithin{definition}{section}
\begin{document}

\title{Universality of Random Matrices and Local Relaxation Flow}

\author{L\'aszl\'o Erd\H os${}^1$\thanks{Partially supported
by SFB-TR 12 Grant of the German Research Council},
Benjamin Schlein${}^2$\;
and Horng-Tzer Yau${}^3$\thanks{Partially supported
by NSF grants DMS-0757425, 0804279} \\
\\
Institute of Mathematics, University of Munich, \\
Theresienstr. 39, D-80333 Munich, Germany${}^1$ \\ \\
Department of Pure Mathematics and Mathematical Statistics
\\  University of Cambridge \\
Wilberforce Rd, Cambridge CB3 0WB, UK${}^2$ \\ \\
Department of Mathematics, Harvard University\\
Cambridge MA 02138, USA${}^3$ \\ \\
\\}

\date{Nov 24, 2010}

\maketitle

\begin{abstract}

Consider the  Dyson Brownian motion with parameter $\beta$, where  $\beta=1, 2, 4$ corresponds to 
 the eigenvalue flows for the eigenvalues of symmetric, hermitian and quaternion 
self-dual ensembles.  For any $\beta \ge 1$, 
we prove that the  relaxation time to local equilibrium for  the Dyson Brownian motion is 
bounded above by $N^{-\zeta}$ for some $\zeta> 0$. The proof is based on an estimate 
of the entropy flow of  the Dyson Brownian motion w.r.t. a ``pseudo 
equilibrium measure".  As an application of this estimate, we prove that 
the eigenvalue spacing statistics in the bulk of the spectrum for
$N\times N$ symmetric  Wigner  ensemble  is the same as 
that of  the Gaussian Orthogonal Ensemble  (GOE) in the limit $N \to \infty$. The 
 assumptions on 
the probability distribution of the matrix elements of the Wigner ensemble are  
a subexponential decay and some  minor restriction on the support.

\end{abstract}


\medskip

{\it Running title:} Universality for Wigner matrices

\medskip


\newpage

\section{Introduction}

A central question concerning random matrices is the universality conjecture
which states that  local statistics  of eigenvalues are determined by
the symmetries of the ensembles
but are otherwise  independent of the details of  the distributions.
There are two types of universalities: the edge universality and the
bulk universality concerning the interior of the spectrum.  The edge
universality is commonly approached via the moment method \cite{SS, Sosh} while 
the
bulk universality was proven for very general classes of unitary invariant 
ensembles
(see, e.g. \cite{BI, DKMVZ1, DKMVZ2, M, PS} and references therein) based on
detailed analysis of
orthogonal polynomials.
The most prominent non-unitary ensembles are the Wigner matrices,
i.e., random matrices with i.i.d. matrix elements that follow a general distribution.
The bulk universality  for  {\it Hermitian} Wigner ensembles
was first established in \cite{EPRSY} for ensembles with smooth distributions.
The later work \cite{TV} by Tao and Vu did not assume smoothness
but it required some moment condition which was removed later in \cite{ERSTVY}.
%Both  approaches of \cite{EPRSY} and  \cite{TV} are based  on
%the following three common ingredients: 
Our  approach  \cite{EPRSY}  to prove the universality was based  on
the following three steps.

\begin{description}
\item[\it Step 1: Local semicircle law. ] It states that the number of eigenvalues in a
spectral  window containing about $N^\e$ eigenvalues is given by the semicircle law 
with a very high probability \cite{ESY1, ESY3}. The factor 
$N^\e$ can be improved to  any sufficiently  large constant
at the expense of deterioriation of the probability estimate.

\item[\it Step 2: Universality for Gaussian divisible ensembles.] The Gaussian divisible ensembles 
are given by matrices of the form 
\be
\wh H+ \sqrt{s} V,
\label{HaV}
\ee
where $\wh H$  is a Wigner matrix,
$V$ is an independent standard GUE matrix
and $s > 0$.  Johansson  \cite{J}  and the later improvement in \cite{BP}  proved 
that the bulk universality holds for
ensembles  of the form \eqref{HaV} if $s > 0$ is independent of $N$. In the work \cite{EPRSY}, 
this result was extended to $s = N^{-1 + \e}$ for any $\e > 0$. 
The key ingredient for this extension was the local semicircle law.

\item[\it Step 3: Approximation by Gaussian divisible ensembles.]  For any given Wigner
 matrix $H$, we find another Wigner matrix $\wh H$ so that
the eigenvalue statistics  of  $H$ and $\wh H+ \sqrt{s} V$ are close to each other.
 The choice of $\wh H$ is given by a  reverse heat flow argument.

\end{description}

Johansson's proof  of the universality of Hermitian Wigner ensembles 
relied  on the asymptotic  analysis of an
explicit formula by Br\'ezin-Hikami  \cite{BH, J} for
the correlation functions of the eigenvalues of $\wh H+ \sqrt s V$.
Unfortunately, the similar formula for GOE is not very explicit
and the corresponding result is not available.
On the other hand, the eigenvalue distribution of the matrix 
$\wh H+ \sqrt s V$ is the same as that of $\wh H + V(s)$, where the matrix elements of $V(s)$ are 
independent standard Brownian motions with variance $s/N$. Dyson  observed that 
the evolution of the eigenvalues of the flow $s\to \wh H + V(s)$ is given by a  system of coupled
stochastic differential equations (SDE), commonly called
the  Dyson Brownian motion (DBM) \cite{Dy}.

If we replace the Brownian motions by the
Ornstein-Uhlenbeck  processes, the resulting dynamics on the eigenvalues,
which we still call
DBM,  has the GUE or GOE  eigenvalue  distributions as the invariant measures
depending on the  symmetry type of the ensembles. 
Thus  the result of Johansson can be interpreted as stating that the
local statistics of GUE is
reached via DBM for time of order one. In fact, by analyzing the dynamics
of DBM with ideas from the hydrodynamical limit,
we have extended Johansson's result  to $s\gg N^{-3/4}$  \cite{ERSY}.
The  key
observation of \cite{ERSY}   is that the local statistics of
eigenvalues  depend exclusively on the approach to local equilibrium.
This method  avoids the usage of explicit formulae for correlation functions,   
but
the identification of  local equilibria, unfortunately,  still uses
explicit representations of correlation functions by  orthogonal polynomials
(following e.g. \cite{PS}),
 and the extension to other ensembles is not a simple task.

Therefore, the universality for symmetric random matrices  remained
open and the only partial result is  Theorem 23 of \cite{TV} 
for Wigner matrices with  the first four moments of the matrix elements
matching  those of GOE. 
The approach of \cite{TV} consisted of three similar steps
as outlined above. For Step 2, it used the result of \cite{J}.
For Step 3, a four moment comparison theorem
for individual eigenvalues was proved in \cite{TV} and the
 local semicircle law (Step 1)
was one of the key inputs in this proof.

 In this paper, we  introduce a general approach to prove local ergodicity 
of DBM, partially motivated by the previous work \cite{ERSY}. 
  In this approach the
analysis of orthogonal polynomials or explicit formulae
are completely eliminated and the method applies to both Hermitian and symmetric 
ensembles.
In fact, the heart of the proof is a convex analysis and it
applies to $\beta$-ensembles for any $\beta \ge 1$. The model specific 
information
required to complete this approach involves only rough estimates on the
accuracy of the local density of eigenvalues.
We expect this method to apply to a very general class of models. More detailed explanations 
will be given in Section \ref{sec3}.

\section{Statement of Main Results}

To fix the notation, we will present the case of symmetric  Wigner matrices; the 
modification to the Hermitian case is straightforward and will be omitted.
The extension to the quaternion self-dual case is also
standard, see, e.g. \cite{ESYY} for the notations and setup.
 On the other hand, the main theorem on DBM  (Theorem \ref{thm:main})
 is valid for  general $\beta$-ensembles. Thus 
all notations for matrices will be restricted to symmetric matrices but  
all results for flows will 
be stated and proved for  general $\beta$-ensembles.
We first explain our general result about  DBM
and in Section \ref{sec:univwigner} we will present its  application 
to Wigner matrices.

\subsection{ Local ergodicity of Dyson Brownian motion}

%We now introduce the Dyson Brownian motion. It is well-known that 
The joint distributions of the
eigenvalues $\bx= (x_1, x_2, \ldots , x_N)$
of the Gaussian Unitary Ensemble (GUE)
and the Gaussian Orthogonal Ensemble (GOE)
are given by the following measure
\be\label{H}
\mu=\mu_N^{(\beta)}(\rd{\bf x})=
\frac{e^{-N\cH({\bf x})}}{Z_\beta}\rd{\bf x},\qquad \cH({\bf x}) =
\left [ \beta \sum_{i=1}^N \frac{x_{i}^{2}}{4} -  \frac{\beta}{N} \sum_{i< j}
\log |x_{j} - x_{i}| \right ],
\ee
where
$\beta=1$ for GOE and $\beta=2$ for GUE. 
We will sometimes use $\mu$ to denote the density
of the measure as well, i.e., $\mu(\bx) \rd\bx = \mu(\rd \bx)$.
We consider $\mu$  defined on the {\it ordered set}
$$
  \Sigma_N : = \{ \bx\in\bR^N \; : \; x_1 < x_2 < \ldots < x_N\},
$$
and  this measure is well-defined for all $\beta  >0 $. 
The Dyson Brownian motion (DBM) is 
characterized by the   
generator
\be
L=   \sum_{i=1}^N \frac{1}{2N}\partial_{i}^{2}  +\sum_{i=1}^N
\Bigg(- \frac{\beta}{4} x_{i} +  \frac{\beta}{2N}\sum_{j\ne i}
\frac{1}{x_i - x_j}\Bigg) \partial_{i}, 
\label{L}
\ee
acting on $L^2(\mu)$.  The DBM is reversible with respect to $\mu$
with the Dirichlet form
\be
D(f) = -\int  f L f  \rd \mu =  \sum_{j=1}^N \frac{1}{2N}
\int (\partial_j f)^2 \rd \mu, 
\label{def:dir}
\ee
where $\partial_j=\partial_{x_j}$.
Notice that we have added a drift $\frac{\beta}{4} x_{i}  \partial_{i}$ 
so that the DBM is reversible w.r.t. $\mu$.  The original definition
by Dyson in \cite{Dy} was slightly different; it contained
no drift term.

%In the original paper 
%by Dyson \cite{Dy}, there is no such  drift term. 

Denote the distribution of the process  at the time $t$
by $f_t ({\bf x})\mu(\rd {\bf x})$.
Then $f_t$ satisfies
\be\label{dy}
\partial_{t} f_t =  L f_t.
\ee
The corresponding stochastic differential equation
for  ${\bf x}(t)$ is now given by  (see, e.g.
Section 12.1 of \cite{G})
\be\label{sde}
 \rd  x_i  =    \frac{\rd B_i}{\sqrt{N}} +  \left [ -  \frac{\beta}{4}
x_i+  \frac{\beta}{2 N}\sum_{j\ne i}
\frac{1}{x_i - x_j}  \right ]  \rd t, \qquad 1\leq i\leq N,
\ee
where $\{ B_i\; : \; 1\leq i\leq N\}$ is a collection of
independent Brownian motions. The  well-posedness of DBM  on $\Sigma_N$
has been proved
in Section 4.3.1 of \cite{G}, see the Appendix for some more  details. 
This  step requires $\beta \ge 1$ which we will assume from now on.

 The dynamics given by \eqref{dy} and \eqref{sde} with  $\beta=1, 2, 4$ 
can be realized by the evolution of the eigenvalues of symmetric,  hermitian and 
quaternion self-dual matrix ensembles, 
but the dynamics is well-defined for $\beta \ge 1$ independently
of the original matrix models.
Our main result, Theorem \ref{thm:main}, is valid 
for all  $\beta\ge1$.

The expectation with respect to the  density $f_t$ will be denoted by $\E_t$ with
$\E:= \E_0$. 
%The expectation with respect to the equilibrium $\mu$ is denoted by $\E^\mu$
For any $k\ge 1$ we define the $k$-point
correlation functions (marginals) of the probability measure $f_t\rd\mu$ by
\be
 p^{(k)}_{t,N}(x_1, x_2, \ldots,  x_k) := \int_{\R^{N-k}}
f_t(\bx) \mu(\bx) \rd x_{k+1}\ldots
\rd x_N.
\label{corr}
\ee
Similarly,  the correlation functions of the equilibrium measure  are denoted by
$$
 p^{(k)}_{\mu,N}(x_1, x_2, \ldots,  x_k) := \int_{\R^{N-k}}
\mu(\bx) \rd x_{k+1}\ldots
\rd x_N.
$$
For the purpose of these definitions only,
 we have extended the measures  $f_t(\bx) \mu(\bx)\rd\bx$ and
 $\mu(\bx)\rd\bx$ from the ordered set $\Sigma_N$ to $\R^N$
by requiring that the resulting measure on $\R^N$  is symmetric w.r.t. 
 all permutations of $\{1, \ldots  N \}$  and 
with a slight abuse of notations,
 we  use the same notation $f_t(\bx) \mu(\bx)$ to denote the density of 
the symmetrized measure. Apart from this definition, we follow the 
convention that the measures are defined on $\Sigma_N$
and all integrations in this paper are on the set $\Sigma_N$ unless otherwise  specified.

Recall that 
$$
  \varrho_{sc}(x) : =\frac{1}{2\pi}\sqrt{(4-x^2)_+}
$$
is the density of the semicircle law and  it is well-known that $ \varrho_{sc}$ is also the density 
w.r.t.  the measure $\mu$ in the limit $N \to \infty$ for $\beta \ge 1$.
Define
\be
 n_{sc}(E) := \int_{-\infty}^E \varrho_{sc}(x) \rd x,
\label{def:Nsc}
\ee
and let $\gamma_j$ be the classical location of the $j$-th eigenvalue
\be\label{gamma}
\gamma_j = n_{sc}^{-1} (j/N).
\ee
Introduce the quantity
\be\label{Q}
Q: =  \sup_{t\ge 0}  N  \int  \left [  \frac 1 {N }  \sum_{j=1}^N
|x_j-\gamma_j|  \right ]^2   f_t(\bx) \rd \mu(\bx).
\ee
Our key result on the local ergodicity of DBM is the following theorem.

\begin{theorem} [Local Ergodicity  of Dyson Brownian Motion] \label{thm:main} 
 Suppose  the initial density $f_0$ satisfies
$S_{\mu}(f_0):=\int f_0\log f_0\rd\mu \le CN^m$ with some fixed exponent $m$
independent of $N$.
Let $f_t$ be  the solution  of the forward equation
\eqref{dy}. 
Suppose that  the following three assumptions are satisfied
for all sufficiently large $N$.
\begin{description}
\item[Assumption 1. ]  For some  $\fa>0$  we have
\be\label{Qbound}
Q \le N^{ - 2 \fa }.
\ee 
\item[Assumption 2. ] There exist  constants $\fb > 0$ and $\fc>0$ such that 
\be\label{sign2} 
\sup_{t \ge 0} 
  \int {\bf 1}  \Big\{ \max_{j=1, \ldots N} |x_j-\gamma_j| 
  \ge N^{-\frak b}\Big\} f_t\rd\mu\le \exp {\big[ -N^{\fc}\big]}.
\ee
\item[Assumption 3. ]
For any compact subinterval $I_0\subset (-2, 2)$,
and for  any $\delta>0$,  $\sigma>0$ and  $n\in \N$ there is a  constant $C_n$ depending on $I_0$,
$\delta$ and $\sigma$ such that for any interval $I\subset I_0$ with
$|I|\ge N^{-1+\sigma}$ and for any $K\ge 1$, we have
\be
   \sup_{t \ge 0}
    \int {\bf 1}\big\{ \cN_I \ge KN|I| \big\}f_t \rd\mu
\le C_{n} K^{-n}.
\label{ass4}
\ee
\end{description}
 Let $E\in \bR$ such that $ |E| < 2$ and let $0<b<2-|E|$.
There exists a constant $\zeta > 0$,  depending only on $\fa$ and $\fb$, such that, 
for any integer $k\ge 1$ and for any compactly supported continuous test function
$O:\bR^k\to \bR$   we have
\be
\begin{split}
\lim_{N\to \infty} \;
\int_{E-b}^{E+b}  \frac { \rd E' } { 2 b} 
\int_{\R^k} &  \rd\alpha_1
\ldots \rd\alpha_k \; O(\alpha_1,\ldots,\alpha_k) \\
&\times \frac{1}{\varrho_{sc}(E)^k} \Big ( p_{t,N}^{(k)}  - p_{\mu, N} ^{(k)} \Big )
\Big (E'+\frac{\alpha_1}{N\varrho_{sc}(E)},
\ldots, E'+\frac{\alpha_k}{ N\varrho_{sc}(E)}\Big) =0
\label{abstrthm}
\end{split}
\ee
for  $ t = N^{-\zeta}$. 
\end{theorem}

{\it Convention.} We will use the letters $C$ and $c$ to denote
general positive constants whose precise values are irrelevant and they
may change from line  to line.

\bigskip

%Theorem \ref{thm:main} states that, if   the point $x_j$ 
%is near its classical location 
%$\gamma_j$ with an error less than $N^{-1/2 - \fa}$ upon  averaging 
% over $j$  in the mean square sense  for all times, 
%then  the DBM reaches the local equilibrium
% at time $N^{-\zeta}$ for some $\zeta > 0$.  
%There are two additional technical
% assumptions. The first one requires that the event
% $|x_j- \gamma_j| \ge N^{-\fb }$ is 
%a large deviation event; the second one requires  
%that the local density is bounded
% on scale up to $N^{-1 + \sigma}$ for all $\sigma > 0$.  

\bigskip

The {\it relaxation time to global equilibrium }
for the DBM is order one in our scaling.  The simplest way to see this is via the 
Bakry-Emery theorem \cite{BE}  which states that, roughly speaking, the relaxation time 
is the inverse of the lower bound to the Hessian of the Hamiltonian $\cH$. In our case
$\cH''\ge I$, and this implies that the relaxation time is order one. 
On the other hand, it was conjectured by Dyson \cite{Dy} that the  relaxation
time to local  equilibrium is of order $N^{-1}$.  Theorem  \ref{thm:main} 
asserts  that the relaxation time to local equilibrium is less than 
$N^{-\zeta}$. 
Although this is far from proving 
Dyson's conjecture,  it is  the first effective estimate that shows
that the local equilibrium is approached much faster than the global one.
Moreover, this result 
 suffices  to prove the bulk universality of Wigner matrices when combining with 
the reverse heat flow ideas introduced in \cite{EPRSY}. 
 We remark that the concept of local equilibrium is  used 
vaguely  here and in Dyson's paper.
In principle, there are many local equilibria depending 
on boundary conditions and the uniqueness 
is a tough question especially now that the interaction is long ranged and singular.

The proof of Theorem \ref{thm:main} is based on the introduction of
 {\it the pseudo equilibrium measure}
which we now explain. It is common  that the global and the local equilibrium are reached 
at different time scales for interacting particle systems, of which DBM is a special case. 
On the other hand,   the hydrodynamical approach 
\cite{ERSY} for the DBM yields very complicated estimates. The main reason for the complications 
is due to that the equilibrium measure of DBM  has a logarithmic two body interaction 
that is both long range and singular at short distances.  Hence the proof of the uniqueness 
 of "local equilibrium measures" is very complicated and 
 we were able to carry it out only for the Hermitian case
due  to that several  identities involving orthogonal
polynomials are valid only for the $\beta=2$ case.
 However, there are two key observations
 from this study: 

\begin{enumerate}
\item The local statistics does not depend on the long range part of the logarithmic interaction,
 in other words, 
we can cutoff the interactions between far away particles 
 without changing   the local statistics.
% are still preserved. 

\item The relaxation time for the gradient flow associated with
the local equilibrium with a fixed boundary condition
is much smaller  than the 
global relaxation time of the DBM, which is of order one.  
\end{enumerate}

To finesse   the difficulty associated with the uniqueness of  local equilibria, 
 we define the {\it  pseudo equilibrium measure}, $\omega$, 
 by cutting off the long  range interactions of the equilibrium measure $\mu$
 and show that $\omega$,  $\mu$ and $f_t \mu$ all have the same local statistics. 
{\it The key idea that the last assertion holds is to 
estimate the relative entropy of the solution to  the DBM, $f_t \mu$, 
relative to the pseudo equilibrium measure $\omega$.}
 Since the pseudo equilibrium measure 
is not a global equilibrium measure, the entropy will not decay monotonically as in the case of 
the relative entropy w.r.t. the equilibrium measure.
More precisely, the time derivative of the relative entropy, under the flow of DBM, 
  w.r.t. the pseudo equilibrium measure  consists of two terms (see Theorem \ref{thm1}): 
 (i) a dissipation term of Dirichlet form of $f_t \mu$ w.r.t.  the pseudo equilibrium measure;
 (ii) an error term  due to the fact that  
the pseudo equilibrium measure  is not the true  equilibrium measure.

Since the  logarithmic interactions between far away particles 
 can be approximated by  a mean-field potential obtained from using the local density, 
the  error term in (ii) can be controlled if we know  the local 
density of particles w.r.t the distribution $f_t \mu$.  The precise
 conditions are  the Assumptions  (1)--(3).   In the special case of $\beta=1, 2 $,
when the DBM is generated by symmetric 
matrix ensembles,
these assumptions will be  verified 
in Lemma \ref{lm:bbound} by using  the local semicircle law; the case of $\beta=4$ is similar and 
some details are given in \cite{ESYY}. 
For other values of $\beta$ it is an open question to verify the corresponding assumptions.
   Given that the error term in (ii) can be bounded, 
we obtain an estimate on the Dirichlet form of $f_t \mu$ w.r.t.  the pseudo equilibrium measure.
The key question is whether this estimate alone is sufficient to pin down the local statistics. 
For this purpose, we note that the Dirichlet form w.r.t. $\omega$ generates a new gradient flow, 
 {\it the local relaxation flow}. 
The {\it global relaxation time} of the local relaxation flow, determined by the convexity of the 
pseudo equilibrium measure, is much shorter than that of  the standard DBM. This leads to 
strong estimates on the local relaxation flow and in particular, it identifies the local statistics. 
The details of the entropy estimates and the local relaxation flow will appear in Section \ref{sec3}.  
We now state the main application of Theorem \ref{thm:main}, 
the universality of symmetric Wigner ensembles.

\subsection{Universality of symmetric Wigner ensemble}\label{sec:univwigner}

{F}ix $N\in\N$ and we consider the  symmetric matrix ensemble
of $N\times N$ matrices $H=(h_{\ell k})$ with
normalization  given by
\be
  h_{\ell k} = N^{-1/2}    x_{\ell k},
\label{scaling}
\ee
where $x_{\ell k}$ for $\ell<k$ are independent,
identically distributed random variables
with the distribution $\nu$ that has
zero expectation and  variance $1$.  The diagonal
elements  $x_{\ell \ell}$ are also i.i.d.  with distribution $\wt \nu$
that has zero expectation and  variance two.
We will assume that
 $\nu$  satisfy the logarithmic Sobolev
inequality (LSI),  i.e. that there is a constant $\theta$ such that for
any nonnegative function
$u$ with $\int  u \; \rd\nu= 1$ we have
\be
  \int u\log u \; \rd\nu \le \theta \int |\nabla\sqrt{u}|^2\rd\nu\; .
\label{logsob}
\ee
We  remark that \eqref{logsob} implies \cite{Le}
that $\nu$ has  a Gaussian decay, i.e. 
\be
  D:= \int e^{\delta_0 x^2}\rd\nu(x) <\infty
\label{gauss}
\ee
for some $\delta_0>0$.  We require that $\wt \nu$ also
satisfies \eqref{logsob}. In this paper, all conditions
and statements involving $\nu$ apply to $\wt\nu$ as well, but
for the simplicity of the presentation,
we will neglect mentioning  $\wt\nu$ all the times.

\medskip

Suppose that the matrix element $h_{\ell k}$ of the symmetric ensembles   evolves according to
the  Ornstein-Uhlenbeck  process  on $\bR$, i.e.
the distribution $\nu_t = u_t   \gamma  $ at the time $t$  satisfies 
\be\label{ou}
\partial_{t} u_t =  A u_t, \quad
  A  = \frac{1}{2}\frac{\pt^2}{\pt x^2}
- \frac{ x}{2} \frac{\pt}{\pt x}
\ee
where $\gamma$ is the standard Gaussian distribution with variance one. 
For the diagonal  element, the Ornstein-Uhlenbeck process should be replaced by the one reversible 
w.r.t. the Gaussian measure with  variance two due to the convention 
that the variances of the diagonal elements are equal to two.  
The Ornstein-Uhlenbeck  process \eqref{ou} induces a
stochastic process on the eigenvalues; it is well-known that 
the process on the eigenvalue is given by the DBM \eqref{sde} with $\beta=1$. 
Notice that we used the Ornstein-Uhlenbeck  process so that the resulting 
DBM is reversible w.r.t. $\mu$.

Our goal is to apply Theorem \ref{thm:main}   with $\beta=1$ and for this purpose, 
we need to verify Assumptions 1-3.  Assumption 3  follows from 
the local semicircle law,  Theorem \ref{thm:semi}, stated later in Section~\ref{sec:best}. 
Assumptions 1 and 2  can be verified if
the measure $\nu$ satisfies the logarithmic Sobolev
inequality \eqref{logsob}.
The precise statement is the following lemma.

\begin{lemma}\label{lm:bbound}
Suppose the assumption \eqref{logsob}
on the distribution $\nu$ of the matrix elements holds.  Then there  are positive
numbers $\fa$, $\fb$ and $\fc$, depending on $\theta$ from \eqref{logsob},
such that 
\eqref{Qbound} and  \eqref{sign2} hold. 
\end{lemma}

From this Lemma, for symmetric  Wigner matrices whose matrix element distributions 
satisfy the LSI,  the assumptions of Theorem \ref{thm:main} are satisfied.
Hence  the correlation functions w.r.t. $f_t\mu$  and the GOE equilibrium 
measure $\mu^{(\beta=1)}_N$ 
are identical  in the large $N$ limit for some $t = N^{-\zeta}$  in the sense that
 \eqref{abstrthm} holds.
 Together with the reverse heat flow argument,
 we have   the following universality 
theorem for local statistics of Wigner ensembles
 whose matrix element distribution is smooth and satisfies  the logarithmic Sobolev inequality.
Denote by $p_N^{(k)}$ 
the correlation functions of the eigenvalues of the symmetric Wigner ensemble. 
Let   $p_{N, GOE} ^{(k)}$ be the correlation functions of the eigenvalues of GOE, i.e., 
the correlation functions of the equilibrium measure $\mu^{(\beta=1)}_N$. It is well-known 
that  $p_{N, GOE} ^{(k)}$ can be  computed explicitly (see, e.g. Section 7 of \cite{M}).

\begin{theorem}\label{main2}
 Suppose the distribution $\nu$ for the matrix elements satisfies the logarithmic
Sobolev inequality
\eqref{logsob}. Assume that  $\nu$ has a positive density $\nu(x)= e^{-U(x)}$ such that 
for any $j$ 
there are constants $C_1, C_2$, depending on $j$, such that
\be
  |U^{( j)}(x)| \le C_1 (1+x^2)^{C_2   }. %\quad  n \ge 1, \quad \nu(x) \le  C_0 e^{- \delta  x^2}
\label{cond1}
\ee
Then for any $|E| < 2 $ and for any $k\ge 1$, we have
\begin{equation}\label{episode}
\begin{split}
\lim_{b \to 0} \lim_{N \to \infty}  \frac{1}{2 b} \int_{E-
b}^{E+b}
&\int_{\R^k} O(\alpha_1,\ldots,\alpha_k)
\\ &\times  \frac{1}{[\varrho_{sc}(E')]^k} \Big ( p_N^{(k)}  - p_{N, GOE} ^{(k)} 
\Big )
 \Big (E'+\frac{\alpha_1}{N\varrho_{sc}(E')},
\ldots, E'+\frac{\alpha_k}{N \varrho_{sc}(E')}\Big)\ \rd\alpha_1
\ldots \rd\alpha_k \rd E' = 0,
\end{split}
\end{equation}
where  $O: \R^k \to \R$ is  an arbitrary continuous, compactly supported 
function.

\end{theorem}
\bigskip

\bigskip
Theorem \ref{main2} is a simple corollary of Theorem \ref{thm:main}
and the method of the reverse heat flow
\cite{EPRSY}.  It will be proved briefly in Section \ref{sec:corfn}.  
Though we stated the universality in terms of correlation functions,
it also holds  for the eigenvalue gap distribution
and we omit the obvious statement (the analogous statement
for the Hermitian case was formulated in Theorem 1.2 of \cite{EPRSY}).

In the following corollary,  by using Theorem 15 of \cite{TV},  we remove all assumptions  from
Theorem  \ref{main2} except for a decay condition and a technical condition
that $\nu$ is supported in at least three points.  This latter technical condition 
was removed in our later paper \cite{EYY2}, where
we  generalized our approach to a broader class of random matrix ensembles.

\begin{corollary}\label{main3}
Suppose the distribution $\nu$ of the matrix elements has mean zero, variance 
one
and  a tail with a subexponential decay,
i.e. it satisfies that
\be\label{subexp}
   \int_\R {\bf 1}(|x|\ge y)\rd \nu(x) \le C\exp\big[ -y^{\frak{q}}\big], 
\qquad\forall y\ge 0,
\ee
for some constants $C, \frak{q} >0$. Assume that $\nu$ is supported in at least three 
points. 
Then the conclusion \eqref{episode} of Theorem \ref{main2} holds.
\end{corollary}

\begin{proof}
Let $m_j$ denote the moments of $\nu$
$$
   m_j  : = \int_\R x^j \rd\nu,
$$
where $m_1=0$ and $m_2=1$. It is easy to check  that $m_4 \ge  m_3^2+ 1$
with the equality holds only for Bernoulli type distribution supported in two 
points.  Due to the condition that $\nu$ is supported on at least three
points, we thus have 
 $m_4 >  m_3^2+ 1$.
For  fixed numbers $m_3, m_4$ satisfying $m_4 > m_3^2+ 1$, there exists 
a  probability measure $\wh\nu$ on $\R$ with density
$e^{-U(x)} $  such that (i) 
the first four moments of $\wh\nu$ match to those of $\nu$, i.e.
\[
\int_\R  x \; \rd \wh\nu  =0, \quad
\int_\R  x^2 \; \rd \wh\nu  =1, \quad
\int_\R  x^j \; \rd \wh\nu  = m_j , \quad j= 3, 4;
\]
(ii) the derivative bounds \eqref{cond1} hold, and (iii) the logarithmic Sobolev inequality 
\eqref{logsob} holds. 
It is easy to argue that such a measure $\wh \nu$ exists.  Consider the space of all  measures satisfying 
\eqref{cond1} with a finite LSI constant.
Since the condition \eqref{cond1} and the  finite LSI constant condition  are preserved under
 small smooth perturbations which are infinite dimensional, there are enough freedom to choose
 perturbations so as to match the first four moments as long as $m_4 >  m_3^2+ 1$.
An elementary  detailed proof of this  fact is given in 
Lemma C.1.  of \cite{EYY}. 
Therefore,  $\wh\nu$ satisfies the assumption of Theorem \ref{main2}  and thus 
\eqref{episode}
holds for the measure $\wh\nu$. Recall that Theorem 15 in \cite{TV}
asserts that the local eigenvalue statistics
for matrices whose matrix element distributions match
up to the  first  four moments
are the same in the limit $N \to \infty$ (strictly speaking, this theorem was 
proved 
only for hermitian matrices, but the parallel version for symmetric ensembles 
holds as well,
see the remark at the end of Section 1.6 in  \cite{TV}). This proves the corollary.
\end{proof}

\section{Pseudo equilibrium 
measure and Entropy Dissipation Estimates} \label{sec3}

The key idea to prove Theorem \ref{thm:main} is an estimate on the time to local equilibrium for 
the DBM. However, to estimate this time to local equilibrium, we need to introduce 
a different flow, the  local relaxation flow, defined as  the gradient flow of the
 {\it pseudo equilibrium measure}. The  pseudo equilibrium measure  is a measure 
which  has the local statistics of 
 the $\beta$ ensemble but has a strong convexity property.  Fix a positive number
 $\eta$ with $N^{-1/6} \ll \eta \ll  1$,
and for the rest of this paper let $\e > 0$ be a
small positive number which we will not specify.  Let $ \gamma_j^\pm : =
\gamma_j \pm  \eta N^{-\e}$
and   define the mean field potential of  eigenvalues
far away from the $j$-th one as
\be\label{Wj}
W_j(x): =  -\frac \beta N  \sum_{k: |k-j| > N \eta} \log (|x - \gamma_k|+\eta)
\quad \text { if }\quad   x\in I_j :=   (\gamma_j^-, \gamma_j^+) 
\ee
where the summation is over all $k\in \{1, 2, \ldots, N\}$ such that
$|k-j|>N\eta$.
For $ x\ge \gamma_j^+$, we extend $W_j$ by
\be
W_j(x) =  W_j (\gamma_j^+) +  W'_j (\gamma_j^+) \Big (x-\gamma_j^+ \Big )+
\frac {W^{\prime \prime}(\gamma_j^+)}  2 \Big (x-\gamma_j^+ \Big )^2
\ee
and similarly for $x \le \gamma_j^-$. In other words,
$W_j$ is just the simplest convex extension
of the function defined by \eqref{Wj} on $I_j$.
This modification will avoid  the singularities at $x= \gamma_k$.  Notice that 
this
is purely a technical device since
we will show in \eqref{xgamma} of Proposition \ref{prop:xg} that the probability of the regime 
$I_j^c$
is negligible in the sense that
\be\label{4.3}
\sum_{j=1}^N\int {\bf 1}(x_j \not \in I_j) f_t\rd \mu\le Ce^{-cN^{\e}}.
\ee

The pseudo equilibrium  measure  $\om_N=\om$ is defined by
\[
\omega= \frac{1}{\wt Z}
\exp \left [- N   \sum_{j=1}^N \left \{  \beta  \frac {x_j^2} 4  +
 W_j (x_j) \right \}  +  \beta \sum_{i<j }
 \log |x_i- x_j| - \frac \beta 2
\sum_i \sum_{j:  |j-i| > N \eta}   \log( |x_i- x_j|+\eta)  \right ]  =: \frac{
e^{-\wt \cH}}{\wt Z}
\]
We can write $\omega = \psi \mu$ with 
\[
\psi= \frac{Z_\beta}{\wt Z}\exp \left [  -
\frac \beta 2  \sum_{i=1}^N \sum_{j:  |j-i| > N \eta}   \log( |x_i- x_j|+\eta)
- N \sum_{i=1}^N  W_i (x_i)  \right ]  .
\]
Recall that  the  relative entropy with respect to a measure 
$\lambda$ is defined by 
\[
S_\lambda(f) = \int f\log f \rd \lambda, \qquad
S_\lambda (f|\psi)=  \int f \log (f/\psi) \rd \lambda
\]
and the Dirichlet form
\be\label{1.11}
D_\lambda (h ) = \frac 1  {  2 N }  \int (\nabla h)^2  \, \rd\lambda
= \frac 1  {  2 N } \sum_{j=1}^N \int (\partial_j
h)^2  \, \rd\lambda \, .
\ee

The local relaxation flow is defined to be  the reversible dynamics w.r.t.
$\omega$ characterized by the generator  $\wt L$ defined by
\be\label{Lt}
\int f \wt L  g \rd \omega = - \frac 1 {2N} \sum_{j=1}^N
\int \partial_j f \partial_j g \rd \omega
\ee
Explicitly,  $\wt L$ is given by  
 \be\label{tl}
\wt L =  L - \sum_{j=1}^N b_j \partial_j, \quad
b_j :=   \frac \beta {N} \sum_{k: |k-j|> N \eta}
 \frac {\mbox{sgn}(x_j-x_k)}{|x_j-x_k|+\eta}
+ W'_j(x_j),
\ee
where
\be
W'_j(x) =  -\frac \beta {N} \sum_{k: |k-j|> N \eta}
\frac {\mbox{sgn} (x-\gamma_k)}
{|x-\gamma_k|+\eta}
\ee
for $ x \in I_j$. Note that for any $k$ with  $|k-j|>N\eta$, we have
$\gamma_k \not\in 2I_j$, where $2I_j$ is the doubling of the interval $I_j$.
Moreover, for $k= j\pm N\eta$ we have $|\gamma_k-\gamma_j|\le
C\eta^{2/3}$ for some constant $C$, and so  $|x-\gamma_k|\le C\eta^{2/3}$
for $x\in I_j$.  Thus we obtain that
$$
 \inf_{x\in I_j} W''_j(x)
\ge c\inf_{|x|\le 2+\eta} \int_{|x-\gamma| \ge C\eta^{2/3}} 
\frac{\varrho_{sc}(\gamma)
\rd \gamma}{|x-\gamma|^2} \ge c\eta^{-1/3}
$$
 with some positive  constant $c$, using $\beta\ge 1$.
Since $W_j$ was defined by a convex extension outside
$I_j$, the same bound holds for any $x$:
\be\label{5.9}
\inf_j \inf_{ x \in \bR}     W_j^{\prime \prime}(x)
 \ge c \eta^{-1/3} ,
\ee
i.e.,  the  mean field potential is uniformly convex with the convexity bound given in \eqref{5.9}.

The  potential $W$ is chosen to satisfy
the two convexity properties:  (\ref{5.9})
and \eqref{convex} and there are  many other possible choices for $W$. For 
example,
without changing the form of $W$ given in \eqref{Wj},
a more natural choice for $\gamma_j$ would be
\[
\gamma_j = \int  x_j  f_0   \rd \mu .
\]
This may somewhat improve the constant in the estimate \eqref{Qbound}, but the analysis
is more complicated and
we will not pursue this choice in this paper.

\subsection{Local Ergodicity  of Dyson Brownian Motion}

The following theorem is  our main result on the local ergodicity of DBM.

\begin{theorem} \label{thmM}
Suppose that $S_\mu(f_0|\psi) \le C N^m$ for   some $m$
fixed. Let 
$ \tau := \eta^{1/3} N^\delta $ for some $\delta>0$.
Define
\be\label{bbound}
 \Lambda:= \sup_{t\le \tau} \sum_{j=1}^N \int b_j^2 f_t \rd \mu.
\ee
Fix $n\ge 1$,
let  $G:\bR^n\to\bR$ be a bounded smooth function with compact
support and define
\be
 \cG_{i,n}(\bx) :=
G\Big( N(x_i-x_{i+1}), N(x_{i+1}-x_{i+2}), \ldots, N(x_{i+n-1}-x_{i+n})\Big).
\label{cG}
\ee
Then  for any $J\subset \{ 1, 2, \ldots , N-n\}$  we have
\be\label{GG}
\Big| \int \frac 1 N \sum_{i\in J} \cG_{i,n}(\bx) f_\tau  \rd \mu -
\int \frac 1 N \sum_{i\in J} \cG_{i,n}(\bx) \rd \mu \Big|
\le C N^{ \delta/2} \Lambda^{1/2}  + Ce^{-cN^\delta}.
\ee
\end{theorem}

\bigskip We emphasize that Theorem \ref{thmM} applies to {\it all $\beta\ge1$ 
ensembles}
and the only assumption
concerning  the distribution $f_t$
is in  (\ref{bbound}). Notice that the first error term becomes large for $\delta$ large,
i.e., if $\tau$ is large.
The first ingredient to prove Theorem \ref{thmM} is the analysis of
the local relaxation flow. The following theorem shows that
 the local relaxation flow satisfies an entropy
dissipation estimate and  its equilibrium measure 
 satisfies the logarithmic
Sobolev inequality.

\begin{theorem} \label{thm2} {\bf (Dirichlet Form Dissipation Estimate)}
Suppose \eqref{5.9} holds.
Consider the equation
\be
\partial_t q_t=\wt L q_t
\label{dytilde}
\ee
with reversible measure $\omega$.
Denote by $R:= \eta^{1/6}$.
Then we have
\be\label{0.1}
\partial_t D_{\omega}( \sqrt {q_t}) \le - C R^{-2} D_{\omega}( \sqrt {q_t}) 
-
\frac{1}{2N^2}  \int    \sum_{ |i-j| \le N \eta}
\frac 1 {(x_i - x_j)^2}     ( \pt_t \sqrt{ q_t} - \pt_j\sqrt {q_t} )^2 \rd 
\omega ,
\ee
\be\label{0.2}
\frac{1}{2N^2} \int_0^\infty  \rd s  \int    \sum_{ |i-j| \le N \eta}
\frac 1 {(x_i - x_j)^2}     (\pt_i\sqrt {q_s} - \pt_j\sqrt {q_s} )^2 \rd \omega
\le D_{\omega}( \sqrt {q_0})
\ee
and the logarithmic Sobolev inequality
\be\label{lsi}
 S_{\omega}(q)\le C R^{2}  D_{\omega}( \sqrt {q})
\ee
with a universal constant $C$.
Thus the relaxation time to equilibrium is of order $R^2= \eta^{1/3}$ and we have
\be\label{Sdecay}
  S_{\omega}(q_t)\le e^{-Ct R^{-2}} S_\omega(q_0).
\ee
\end{theorem}

The notation $R= \eta^{1/6}$ was introduced so that this result and 
Theorem 4.2  in \cite{ESYY}  are identical. The scale parameter $R$ has 
 a meaning in \cite{ESYY}, but it is purely a choice of convention here. 
The proof given below  follows the
argument in \cite{BE} and it was outlined in this
context in Section~5.1 of \cite{ERSY}.  The new observation is the additional
second term  on the r.h.s of \eqref{0.1},  corresponding to  ``local
Dirichlet form dissipation''.
The estimate \eqref{0.2} on this additional term will
play a key role in this paper.

\medskip

\begin{proof} In 
\cite{ERSY} 
it was shown  that,
with the notation $h=\sqrt{q}$, we have
$$
  \pt_t h = \wt Lh + \frac{1}{2N}h^{-1} (\nabla h)^2
$$
and
\be\label{1.5}
  \pt_t \frac{1}{2N}\int (\nabla h)^2 e^{-\wt\cH} \rd \bx
  \le  - \frac{1}{2N^2}\int
\nabla h (\nabla^2 \wt\cH)\nabla h  e^{-\wt\cH} \rd\bx.
\ee
In our case, \eqref{5.9} and the fact that
$$
   \frac{\rd^2}{\rd x^2} \Big( \log (|x|+\eta) - \log |x| \Big) \ge 0, \qquad 
x>0,
$$
imply that the Hessian of $\wt \cH$ is bounded
from below as
\be\label{convex}
\frac{1}{2N^2}
\nabla h (\nabla^2 \wt\cH)\nabla h
\ge   C \eta^{-1/3}    \frac 1 N\sum_j  (\partial_j h)^2 +
\frac{1}{2N^2}  \sum_{ |i-j| \le N \eta}  \frac 1 {(x_i - x_j)^2} (\pt_i h -
\pt_j h)^2
\ee
with some positive constant $C$.
This proves \eqref{0.1} and \eqref{0.2} since $R^2 = \eta^{1/3}$. 
Inserting the inequality
$$
\partial_t D_{\omega}( \sqrt {q_t}) \le - C R^{-2} D_{\omega}( \sqrt {q_t})
$$
from \eqref{0.1} into the equation
\be\label{1.2}
\partial_t   S_{\omega}(q_t) = -   4D_{\omega}( \sqrt {q_t}),
\ee
and integrating the resulting equation, we prove \eqref{lsi}.
Inserting  \eqref{lsi} into \eqref{1.2} we have
$$
   \pt_t  S_{\omega}(q_t) \le - C R^{-2} S_{\omega}(q_t)
$$
and we obtain \eqref{Sdecay}.
\end{proof}

\begin{remark}\label{R1} The proof of \eqref{1.5} requires an integration
by parts and the boundary term at $x_i=x_j$ (explained
in Section 5.1. of \cite{ERSY}) should vanish. In the
Appendix we will justify this technical step.
\end{remark}

\begin{lemma}
Suppose that the density $q_0$ satisfies $S_\omega(q_0)\le CN^m$ with some $m>0$ 
fixed.
For a fixed $n\ge1$ let  $G:\bR^n\to\bR$ be a bounded smooth function with 
compact support
and recall the definition of $\cG_{i,n}$ from \eqref{cG}.
Let $J\subset \{ 1, 2, \ldots , N-n\}$ and set $\tau = \eta^{1/3} N^{\delta}$.
Then we have
\be\label{diff}
\left | \int \frac 1 N \sum_{i\in J} \cG_{i,n}(\bx)\,\rd \omega -
\int \frac 1 N \sum_{i\in J} \cG_{i,n}(\bx) \, q_0 \rd \omega
\right |
\le C\sqrt { \frac { D_\omega(\sqrt {q_0}) \tau  } {N} } +  Ce^{-cN^{\delta}}
\ee
with some constant $C$ depending only on $G$.
\end{lemma}

\begin{proof}   Without loss of generality, we consider only
the case $J = \{1, \ldots, N-n\}$.   Let $q_t$ satisfy
\[
\partial_t q_t = \wt L q_t
\]
with an initial condition $q_0$.
We first compare $q_\tau$ with $q_\infty=1$.
Using the entropy inequality,
$$
  \int |q-1|\rd\omega \le 2 \sqrt{S_\omega(q)},
$$
and the exponential decay
of the entropy \eqref{Sdecay}, we have
\[
\Big|\int \frac 1 N \sum_{i\in J}  \cG_{i,n}(\bx)\,  q_\tau \rd \omega -
\int \frac 1 N \sum_{i\in J}  \cG_{i,n}(\bx)\, \rd \omega\Big|
\le C \big ( N^m  e^{-\tau \eta^{-1/3}}\big)^{1/2} \le C e^{-cN^\e}.
\]

To compare $q_0$ with $q_\tau$,
by differentiation, we have
\[
\int \frac 1 N \sum_{i\in J} \cG_{i,n}(\bx)q_\tau \rd \omega -
\int \frac 1 N \sum_{i\in J} \cG_{i,n}(\bx)q_0 \rd \omega 
\qquad\qquad\qquad\qquad\qquad\qquad\qquad
\]
\[
\qquad\qquad\qquad = \int_0^\tau  \rd s \int  \frac 1 N \sum_{i\in J} \sum_{k=1}^n
  \pt_k G\Big( N(x_i-x_{i+1}), \ldots, N(x_{i+n-1}-x_{i+n})\Big)
[\pt_{i+k-1} q_s - \pt_{i+k}q_s]  \rd \omega.
\]
{F}rom the Schwarz inequality and $\pt q = 2 \sqrt{q}\pt\sqrt{q}$
the last term is bounded by
\begin{align}\label{4.1}
2\sum_{k=1}^n & \left [   \int_0^\tau  \rd s \int
\sum_{i\in J} \Big[\pt_k G \Big(N(x_i - x_{i +1}),\ldots, N(x_{i+n-1}-x_{i+n}) 
\big)\Big] ^2
(x_{i+k-1}-x_{i+k})^2  \, q_s \rd \omega
\right ]^{1/2} \nonumber \\
& \times \left [ \int_0^\tau  \rd s \int  \frac 1 {N^2 } \sum_{i\in J}
\frac{1}{(x_{i+k-1}-x_{i+k})^2}  [ \pt_{i+k-1}\sqrt {q_s} -
\pt_{i+k}\sqrt {q_s}]^2  \rd \omega \right ]^{1/2} \nonumber \\
\le &  \; C_n \sqrt { \frac{D_\omega(\sqrt {q_0}) \tau}{N}},
\end{align}
where we have used \eqref{0.2} and that
$$
 \Big[\pt_k G \Big(N(x_i - x_{i +1}), \ldots, N(x_{i+k-1}-x_{i+k}), \ldots
N(x_{i+n-1}-x_{i+n}) \Big) \Big]^2
(x_{i+k-1}-x_{i+k})^2 \le CN^{-2},
$$
since  $G$ is smooth and  compactly
supported.
This proves the Lemma.
\end{proof}

\bigskip
Notice if we use only the entropy dissipation and Dirichlet form,
the main term on the right hand side of \eqref{diff} will become $C\sqrt { S_\om(q)\tau}$.
Hence by exploiting the Dirichlet form dissipation
coming from the second term on the r.h.s. of \eqref{0.1},
we gain the crucial  factor $N^{-1/2}$
in the estimate.

\bigskip

The second  ingredient to prove Theorem \ref{thmM} is the following entropy and
Dirichlet form estimates.

\begin{theorem}\label{thm1} {\bf (Entropy and Dirichlet Form Estimates)}
 Suppose the assumptions
of Theorem \ref{thmM} hold. Recall that $\tau = \eta^{1/3} N^\delta$ and
define $g_t := f_t/\psi$ so that
$S_{\mu} (f_t|\psi) =  S_{\omega} (g_t)$.
Then the entropy and the Dirichlet form satisfy the estimates:
\be\label{1.3}
 S_{\omega} (g_{\tau/2}) \le
  C N R^2  \Lambda,
\ee
\be\label{1.4}
D_\omega (\sqrt{g_\tau})
 \le CN \Lambda.
\ee
\end{theorem}

\begin{proof}  First we need  the following  relative entropy identity from 
\cite{Y}.

\begin{lemma}
Let $f_t $  be a probability
density satisfying $\partial_tf_t=Lf_t$.  Then for any probability density
$\psi_t$  we have
\be
\partial_t  S_\mu(f_t|\psi_t) = - \frac 2 {  N } \sum_{j} \int (\partial_j
\sqrt {g_t})^2  \, \psi_t \, \rd\mu
+\int g_t(L-\partial_t)\psi_t \, \rd\mu \ ,
\label{entrdif}
\ee
where $g_t = f_t/\psi_t$.
\end{lemma}

In our setting, $\psi$ is independent of $t$ and $L$  satisfies \eqref{tl}.
 Hence we 
have
$$
\pt_t S_\omega(g_t)=\partial_t  S_\mu(f_t|\psi)
= - \frac 2 {  N } \sum_{j} \int (\partial_j
\sqrt {g_t})^2  \, \rd\omega
+\int   \wt L g_t  \, \rd\omega+  \sum_{j} \int  b_j \partial_j
g_t   \, \rd\omega.
$$
Since the middle term on the right hand side vanishes, we
have from the Schwarz inequality
\be\label{1.1}
\pt_t S_\omega(g_t) \le  -D_{\omega} (\sqrt {g_t})
+   CN  \sum_{j} \int  b_j^2  g_t   \, \rd\omega.
\ee
Together with the LSI \eqref{lsi}  and \eqref{bbound},  we have
\be\label{1.2new}
\partial_t  S_\omega(g_t) \le  - CR^{-2}  S_\omega(g_t) +
  C   N  \Lambda
\ee
for $t\le \tau$.
Since $S_\omega(g_0)=S_\mu(f_0|\psi) \le CN^m$ and $\tau/2 \gg R^{2}$,
 the last inequality  proves \eqref{1.3}.

Integrating  \eqref{1.1}
from $t=\tau/2$ to $t=\tau$ and using
the monotonicity of the Dirichlet form in time, we have proved \eqref{1.4}
with the choice of $\tau$.

\end{proof}

\bigskip

{\bf Proof of Theorem \ref{thmM}.}  Fix $\tau= R^2 N^\delta = \eta^{1/3} N^\delta$
and let $q_0:= g_\tau = f_\tau/\psi$ with
$f_0$ satisfying  the assumption of Theorem \ref{thm1}, i.e., 
$S_\mu(f_0|\psi) \le C N^m$ for   some $m$ and \eqref{bbound} holds. 
Using  \eqref{1.4}, we 
have
\[
\sqrt { \frac { D_{\omega}(\sqrt {q_0}) \tau } N } \le
 CN^{ \delta/2}  \Lambda^{1/2},
\]
and from \eqref{diff} we also have
\be\label{6.1}
\Big | \int \frac 1 N \sum_{i\in J} \cG_{i,n}(\bx) f_\tau  \rd \mu -
\int \frac 1 N \sum_{i\in J} \cG_{i,n}(\bx) \rd \omega \Big | \le
C N^{ \delta/2}  \Lambda^{1/2} + Ce^{-cN^\delta}.
\ee
Clearly, equation  \eqref{6.1}  also holds for the special choice
 $f_0= 1$ (for which $f_\tau=1$), i.e. local statistics of $\mu$ and
$\om$ can be compared. Hence we can replace the measure
$\omega$ in \eqref{6.1} by $\mu$
and this proves Theorem \ref{thmM}.
\qed

\bigskip

%
%%Now we state the universality of the local eigenvalue statistics for Wigner matrices. 
%To formulate such a result in terms of
%correlation functions, it is convenient to remove the ordering
%$x_1<x_2 <\ldots < x_N$ among the eigenvalues and  work with symmetric 
%densities.
%Let $p_N(x_1, x_2, \ldots , x_N)$ denote
%the (symmetric) probability  density of eigenvalues and for any
%$k=1,2,\ldots, N$ let
%\be
% p^{(k)}_N(x_1, x_2,\ldots x_k):=
%\int_{\bR^{N-k}} p_N(x_1, x_2, \ldots , x_N)\rd x_{k+1}\ldots \rd
%x_N
%\label{corrfn}
%\ee
%be the $k$-point correlation function.
%

\section{Proof of Theorem \ref{main2}
 and Theorem \ref{thm:main}}\label{sec:corfn}

 We first prove Theorem \ref{main2} assuming that  Theorem \ref{thm:main} holds. 
Our main tool is the reverse heat flow argument from \cite{EPRSY}. 
Recall that the distribution of the matrix element is given by a measure $\nu $
and the generator of the Ornstein-Uhlenbeck process is  $ A$ \eqref{ou}. 
The probability  distribution of all  matrix elements is 
 $ \nu^{\otimes n}$, $n = N^2 $.
The
joint probability distribution of the matrix elements  at  time $t$ as every matrix element 
evolves under the Ornstein-Uhlenbeck process is 
given by 
\[
F_t \rd \gamma^{\otimes n}:=
(e^{t A} u)^{\otimes n} \; \rd\gamma^{\otimes n},
\]
where we recall that $\gamma$ is the standard Gaussian measure.
%Strictly speaking, the variance of the diagonal element is equal to two 
%and for these elements we have to use a slightly different Ornstein-Uhlenbeck generator
%which would complicate the notation of the product measure. For the sake
%of simplicity, we neglect this notational complication.

\begin{proposition}\label{meascomp1}  Fix a positive integer $K$.
Suppose that $\nu = u \rd\gamma $ satisfies the subexponential decay condition
\eqref{subexp}
and the regularity condition \eqref{cond1} for  all  $ j \le K$. 
 Then there is a small constant $\alpha_K$, depending only on $K$, such
that for  any positive  $t  \le  \alpha_K $
there exists a probability density $g_t$ w.r.t. $\gamma$ with mean zero and variance one such that
\be\label{gtilde}
\int  \left  | e^{t A}g_t  - u \right |
\rd \gamma  \le C\; t^{K}
\ee
for  some $C>0$ depending on $K$.  Furthermore,  $g_t$ can be chosen such that
if  the logarithmic Sobolev inequality \eqref{logsob} holds for the measure $\nu=u\gamma$, then
it holds for  $g_t\gamma$  as well, with   the logarithmic Sobolev constant
changing by a factor of at most $2$.

Furthermore, let  $\cA=A^{\otimes n}$, $ F=  u^{\otimes n}$ with  $n = N^2$ and
set $G_t:=  g_t^{\otimes n}$. Then we also have
\be\label{FFtilde}
\int  \left  | e^{t\cA}G_t  - F \right |
\rd \gamma^{\otimes n}   \le C\; N^2 t^{K}
\ee
for  some $C>0$ depending on $K$.
\end{proposition}

\begin{proof}
This proposition can be proved following the reverse heat flow idea from \cite{EPRSY}.
% Roughly speaking, 
%we can construct $g_t$ by 
%\[
%g_t (x)  = 
%\left [ 1-t A+\frac{1}{2}t^2 A^2  + \ldots +  (-1)^{K-1} \frac
%{t^{K-1}} { (K-1)!}  A^{K-1}  \right ] u (x) , 
%\]
%i.e.,  $g_t \approx e^{- t A} u$, where the exponential is expanded
%up to order $t^K$. The condition \eqref{cond1} assures that the higher 
%order terms are 
%small as long as $x$ is not too large depending on $t$. In this perturbative region, we can check 
%\eqref{gtilde} easily. The statement on the logarithmic Sobolev inequality is also easy to verify since
% $1/2 \le  g_t(x)/u(x) \le 2$.    Some simple cutoff argument is needed for large $x$ regime, 
%but this was carried  out
% in \cite{EPRSY} (see also  \cite{ESYY}).  So we omit the details here. Finally, \eqref{FFtilde}
%is a simple consequence of \eqref{gtilde} and this concludes the proof of
%the proposition.\qed
 Define $
\theta(  x) =   \theta_0 (  t^{\alpha} x)$ with some  small positive $\alpha>0$
depending on $K$,
where $\theta_0$ is a smooth cutoff function satisfying
$\theta_0(x) = 1$ for $|x|\le 1$ and $\theta_0(x) = 0$ for $|x| \ge 2$.
Set
\[
h_s =    u  +  \theta \xi_s   , \quad \mbox{with} \quad
\; \xi_s:=
\left [ -sA+\frac{1}{2}s^2A^2  + \ldots +  (-1)^{K-1} \frac
{s^{K-1}} { (K-1)!}  A^{K-1}  \right ] u .
\]
By assumption \eqref{cond1},  $h_s$ is positive
and
\be\label{ul}
\frac{2}{3} u  \le h_s \le \frac{3}{2} u.
\ee
for any $s\le t$ if $t$ is small enough. 
% To see this, take, e.g., $K=2$  and we have 
%\[
%|\theta(  x) \xi_s(x)|  \le C  s  \theta_0 (  t^{\alpha} x) \Big [\,   \big |V''(x) \big | +  \big |x V'(x) \big 
%| \, \Big ]   \, | u(x)| \le  \frac{1}{2} | u(x)|,
%\]
%where we have used $\alpha \ll 1$, $ s \le t$ and the assumption \eqref{cond1}.

Define $  v_s = e^{s A}h_s$ and by definition, $v_0= u$.  Then
$$
\pt_s v_s =   (-1)^{K-1} \frac   {s^{K-1}} { (K-1)!}   e^{s A} A^{K} u
+ e^{s A} A   (\theta -1)  \xi_s + e^{sA} (\theta-1)\pt_s\xi_s.
$$
Since the Ornstein-Uhlenbeck is a contraction in $L^1(\rd\gamma)$,
together with \eqref{cond1}, we have
\be
\int  |v_t - u|  \rd \gamma  \le C_K   \int_0^t     \int  \Big[ t^{K-1}| A^{K} u |  +
|  A   (\theta -1)  \xi_s | +  |(\theta-1)\pt_s\xi_s|  \Big] \rd \gamma \; \rd s
\le  C_K t^{K}
\label{uut}
\ee
for sufficiently small $t$.
%To estimate the last two terms, we also used that on the support of
%$\theta-1$ the measure $\rd\gamma$ decays subexponentially in $t$.

Notice that $h_t$ may not be normalized as a probability density w.r.t. $\gamma$
but  it is easy to check that there is a constant $c_t = 1 + O(t^M)$, for any $M > 0$ positive,
such that $c_t h_t$ is a probability density. Clearly,
\[
\alpha_t:= \int x c_t h_t \rd \gamma = O(t^M),   \qquad
\sigma^2_t:=   \int (x-\alpha_t)^2 c_t h_t \rd \gamma = 1+ O(t^M),
\]
and the same formulas hold if $h_t$ is replaced by $v_t$ since
the OU flow preserves expectation and variance.
Let $g_t$ be defined by
\[
g_t(x)  e^{- x^2/2} =  c_t  \sigma^{-1}_t   h_t( (x+\alpha_t) \sigma_t^{-1} )
e^{-(x+ \alpha_t)^2/2 \sigma_t^{2}}  .
\]
Then $g_t$ is a probability density w.r.t. $\gamma$ with zero mean and variance $1$.
It is easy to check that the total variation norm of $h_t-g_t$ is
smaller than any power of $t$. Using again the contraction property of $e^{tA}$ and
\eqref{uut}, we get
\be
\int  |e^{t A} g_t - u|  \rd \gamma \le \int  |e^{t A} g_t - e^{t A} h_t |  \rd \gamma + \int  |v_t - u |  \rd \gamma
\le  C t^{K}
\label{uut1}
\ee
for sufficiently small $t$.

Now we  check the LSI constant for $g_t$.  Recall that $g_t$ was
obtained from $h_t$ by translation and dilation. By\ definition of the LSI
constant, the translation does not change  it.
The dilation changes the constant, but since our dilation constant is nearly one,
the change of LSI constant is also nearly one.
So we only have to compare the LSI constants between $\rd\nu = u\rd\gamma$ and $
c_t   h_t \rd \gamma  $. From \eqref{ul} and that $c_t$ is nearly one,
the LSI constant changes by a factor less than $2$.
This proves the claim on the LSI constant.

Finally, the \eqref{FFtilde} directly follows from
$$
\int  \left  | e^{t\cA}G_t  - F \right |
\rd \gamma^{\otimes n} \le N^2 \int  \left  | e^{t A }g_t  - u \right |\rd\gamma
$$
and this completes the proof of Proposition \ref{meascomp1}.
\qed

\medskip

We now apply Theorem \ref{thm:main}  to the  initial distribution given by the 
eigenvalues of the symmetric Wigner ensemble with distribution $g_{\tau} \gamma$
where $\tau = N^{-\zeta}$. By Proposition \ref{meascomp1}, the LSI constant of 
$g_{\tau} \gamma$ is bounded by the initial LSI constant of $\nu$  by a factor of at most two. 
Thus we can apply Lemma \ref{lm:bbound} to verify  Assumptions 1 and 2  
 of Theorem \ref{thm:main}. 
Assumption 3 follows from the local semicircle law, Theorem \ref{thm:semi}.
 Thus the correlation functions 
of the eigenvalues of the ensemble  with distribution $ (e^{\tau A} g_{\tau}) \gamma$ are 
the same as those of GOE in the  sense of \eqref{abstrthm}. Finally, using \eqref{FFtilde},
 we can   approximate the $k$-point correlation function 
w.r.t. $(e^{\tau A} g_{\tau}) \gamma$ by the one w.r.t. $\nu$
 after choosing $K$ sufficiently large so that $N^2 \tau^K = N^{2-K\zeta}=o(N^{-k})$.
The additional smallness factor $N^{-k}$ for the estimate
on the total variation in \eqref{FFtilde} is necessary to conclude
the convergence of the $k$-point correlation function, since it is rescaled
by a factor $N$ in each variable. We also used the trivial fact that
the total variation distance of two eigenvalue distributions is
bounded by the total variation distance of the  distributions
of the full matrix ensembles.  Finally  we remark that
the $b\to 0$ limit in   \eqref{episode} is needed to replace
$\varrho_{sc}(E)$ in \eqref{abstrthm} with $\varrho_{sc}(E')$
in \eqref{episode} using the continuity of $O$. 
  This concludes the proof
 of Theorem \ref{main2}. 
\end{proof}

\bigskip

{\bf Proof of Theorem \ref{thm:main}.}
{\it Step 1.} The first step   is to show that   the right hand side of \eqref{GG}
vanishes in the large $N$ limit for $\eta = N^{-\e_3}$ with $\e_3$ small enough  provided that 
the estimates  (\ref{Qbound}),  (\ref{sign2}) hold.
 By \eqref{sign2},  $x_j \in I_j$ (recall the definition
of $I_j$ from \eqref{Wj}) with a very high probability.  
In this paper we will say that an event holds with 
a very high probability if the complement event 
has a probability that is subexponentially small in $N$,
i.e., it is bounded by $C\exp(-N^\e)$ with some fixed $\e>0$.
{F}rom the definition of $b_j$ \eqref{tl},
we have
\be\label{bj}
b_j  =   \frac \beta N \sum_{k \; : \; |k-j|> N \eta}
\left [ \frac {\mbox{sgn}(x_j-x_k)} {|x_j-x_k|+\eta}
- \frac {\mbox{sgn}(x_j -\gamma_k) }{ |x_j -\gamma_k|+\eta} \right ].
\ee
Notice that function $g(x) :=\frac{\mbox{sgn(x)}}{|x|+\eta}$ satisfies
\be
   |g(x)-g(y)|\le \eta^{-2}|x-y|
\label{gb}
\ee
as long as $x$ and $y$ have the same sign.
In our case, $x_j-x_k$ and $x_j - \gamma_k$
have the same sign  as long as 
\be\label{sign}
|x_k-\gamma_k|< |x_k-x_j|, \qquad \text{for all $ k$ satisfying  $|k-j|> N \eta$.}
\ee
The last inequality holds with a very high probability due to \eqref{sign2} 
provided $\e_3$ is smaller than $\fb$.
We remark that this is the only place where we used Assumption 2.
Thus,,  with a very high probability, we have 
\be\label{bj2}
|b_j|  \le   \eta^{-2}  \frac \beta N \sum_{k \; : \; |k-j|> N \eta}
|x_k -\gamma_k| .
\ee
The contribution to $\Lambda$ of the exceptional event
is negligible, since its probability is
subexponentially small in $N$ and $|b_j|\le C\eta \le C N^{\e_3}$.
Thus, recalling the definition of $Q$ from \eqref{Q} and the definition
of $\Lambda$ from \eqref{bbound}, 
 we can bound the error term on the right hand side of \eqref{GG}  by 
\be\label{bbound2}
N^{\delta/2} \Lambda^{1/2}   \le   CN^{\delta/2} Q^{1/2}\eta^{-2}
   \le C N^{- \fa  +   2\e_3 + \delta/2 }
\le N^{- \fa/2}\to 0,
\ee
provided that \eqref{Qbound} holds and 
 $\e_3$ and $\delta$ are small enough, depending on $\fa$.

\medskip

{\it Step 2. From \eqref{GG} to correlation functions}.  The equation \eqref{GG} shows
 that for a special class of observables, depending only 
on rescaled {\it differences} of the points $x_j$, the expectations w.r.t. $f_t \mu$ and 
w.r.t the 
equilibrium measure $\mu$ are identical in the large $N$ limit. 
But the class of observables in \eqref{abstrthm} of
Theorem~\ref{thm:main} is somewhat bigger and we need to
 extend our result to them. Without the $\rd E'$ integration in 
\eqref{abstrthm}, the observable would strongly depend on
a fixed energy $E'$ and could not be approximated by observables depending
only on differences of $x_j$. Taking a small averaging in $E'$ remedies
this problem.

We will consider $E$, $b$ and $n$ fixed, i.e., the constants
in this proof may depend on these three parameters. We start with the identity
\begin{align}
\int_{E-b}^{E+b}  \frac { \rd E'} { 2 b}  \int_{\R^n} \rd\alpha_1
\ldots \rd\alpha_n  & \; O(\alpha_1,\ldots,\alpha_n)  p_{\tau, N}^{(n)}
\Big (E'+\frac{\alpha_1}{N\varrho(E)},
\ldots, E'+\frac{\alpha_n}{N \varrho(E)}\Big) \label{iddd} \\
 = & \int_{E-b}^{E+b}  \frac { \rd E'} { 2 b}   \int \sum_{i_1\ne i_2\ne
 \ldots \ne i_n}
  \wt O\big( N(x_{i_1}-E'),  N(x_{i_1}-x_{i_2}),
 \ldots  N(x_{i_{n-1}}-x_{i_n})\big)
 f_\tau\rd\mu, \nonumber
\end{align}
where $\wt O (u_1, u_2, \ldots u_n): = 
O\big( \varrho(E)u_1, \varrho(E)(u_1-u_2), \ldots\big)$. By permutational symmetry of
$ p_{\tau, N}^{(n)}$ we can assume that $O$ is symmetric and we can restrict
the last summation to $i_1 < i_2 < \ldots < i_n$ upon an overall factor $n!$.
Let $S_n$ denote the set of  
increasing positive integers, ${\bf m} = (m_2, m_3, \ldots, m_n) \in \N_+^{n-1}$, 
$m_2< m_3 <\ldots < m_n$.
For a given $\bm\in S_n$, we
change the indices to $i_1=i$, $i_2= i+m_2$, $i_3=i+m_3, \ldots,$
 and rewrite
the sum on the r.h.s. of \eqref{iddd} as
\begin{align}
  \sum_{{\bf m}\in S_n} \sum_{i=1}^N & \wt O\big(  N(x_i - E'),
N(x_i - x_{i+m_2}),
  N(x_{i+m_2} - x_{i+m_3}) , \ldots \big) =  \sum_{{\bf m}\in S_n}\sum_{i=1}^N  Y_{i,\bm}
(E',\bx),
\non 
\end{align}
where we introduced
$$
   Y_{i,\bm}(E',\bx): = \wt O\big(  N(x_i - E'),
N(x_i - x_{i+m_2}) ,\ldots,  N(x_i- x_{i+m_n})   \big).
$$
We will set $Y_{i,\bm}=0$ if $i+m_n>N$.
We have to show that
\be
   \lim_{N\to\infty}
  \Bigg| \int_{E-b}^{E+b}  \frac { \rd E'} { 2 b} 
   \int \sum_{{\bf m}\in S_n}\sum_{i=1}^N  Y_{i,\bm}
(E',\bx) f_\tau\rd \mu - 
 \int_{E-b}^{E+b}  \frac { \rd E'} { 2 b}   \int \sum_{{\bf m}\in S_n}\sum_{i=1}^N  Y_{i,\bm}
(E',\bx) \rd \mu\Bigg|=0.
\label{goal}
\ee
Let $M$ be an $N$-dependent parameter chosen at the end of the proof. Let 
$$
  S_n(M): = \{ \bm \in S_n \; , \; m_n \le M\}, \quad S_n^c(M):= S_n\setminus S_n(M),
$$
and note that $|S_n(M)|\le M^{n-1}$. 
To prove \eqref{goal}, it is sufficient to show that
\be
   \lim_{N\to\infty}
  \Bigg| \int_{E-b}^{E+b} \frac { \rd E'} { 2 b}   
 \int \sum_{{\bf m}\in S_n(M)}\sum_{i=1}^N  Y_{i,\bm}
(E',\bx) f_\tau\rd \mu - 
 \int_{E-b}^{E+b}  \frac { \rd E'} { 2 b}  \int \sum_{{\bf m}\in S_n(M)}\sum_{i=1}^N  Y_{i,\bm}
(E',\bx) \rd \mu\Bigg|=0
\label{goal1}
\ee
and that
\be
   \lim_{N\to\infty}  \sum_{{\bf m}\in S_n^c(M)}
  \Bigg| \int_{E-b}^{E+b}  \frac { \rd E'} { 2 b}    \int \sum_{i=1}^N  Y_{i,\bm}
(E',\bx) f_\tau\rd \mu \Bigg|=0
\label{goal2}
\ee
hold for any $ \tau  \ge  \eta^{1/3} N^{\delta}$ (note that $\tau =\infty$ 
corresponds to the equilibrium, $f_\infty =1$), where $  \eta^{1/3} N^{\delta}$ is chosen in 
Theorem \ref{thm:main} and $\eta$ is chosen in the Step 1. 
\bigskip

\noindent
{\it Case 1: Small $\bm$ case; proof of \eqref{goal1}.}

\bigskip
 After performing the $\rd E'$ integration, we will eventually
apply Theorem \ref{thmM} to the function
$$
G\big( u_1, u_2, \ldots \big)
 : = \frac{1}{2b} \int_{\R} \wt O\big( y,
 u_1, u_2 ,\ldots,   \big) \rd y ,
$$
i.e., to the quantity  
\be
 \int_\R   \frac { \rd E'} { 2 b}   \; Y_{i,\bm}(E',\bx)= 
%   \int_{\R}\rd E'  \wt O\big(  N(x_i - E'),N(x_i - x_{i+m_2}) ,\ldots,   \big)  =
 \frac{1}{N} G\Big( N(x_i-x_{i+m_2}), \ldots \Big) 
\label{OO}
\ee
for each fixed $i$ and ${\bf m}$.

For any $E$ and  $0<\xi<b$ define sets of integers
$J=J_{E,b,\xi}$ and $J^\pm= J^\pm_{E,b,\xi}$  by
$$
  J : = \big\{ i\; : \; \gamma_i \in [E-b, E+b]\big\},\quad
  J^\pm : = \big\{ i\; : \; \gamma_i \in [E-(b\pm\xi), E+b\pm\xi]\big\},
$$
where $\gamma_i$ was defined in \eqref{gamma}.  Clearly $J^-\subset J \subset J^+$.
With these notations, we have
\be
   \int_{E-b}^{E+b} \frac { \rd E'} { 2 b}   \sum_{i=1}^N
 Y_{i,\bm}(E',\bx )=  \int_{E-b}^{E+b}\frac { \rd E'} { 2 b}    \sum_{i\in J^+} Y_{i,\bm}(E',\bx )
 + \Omega^+_{J,\bm}(\bx).
\label{uppe}
\ee
 The error term $\Omega^+_{J,\bm}$, defined by \eqref{uppe}
indirectly, comes from  those $i\not\in J^+$ indices,
for which $x_i \in [E-b, E+b] + O(N^{-1})$ since 
$Y_{i,\bm}(E',\bx)=0$ unless $|x_i-E'|\le C/N$, the constant
depending on the support of $O$. Thus
$$
   |\Omega^+_{J,\bm}(\bx)| \le  CN^{-1}\# \{ \; i \; : \; |x_i-\gamma_i|\ge \xi/2 \}
$$
for any sufficiently large $N$ assuming $\xi\gg 1/N$
and using that $O$ is a bounded function. The additional $N^{-1}$ factor
comes from the $\rd E'$ integration. 
 Taking the expectation with respect to the
measure $f_\tau\rd\mu$, and by a Schwarz inequality, we get
\be
    \int |\Omega^+_{J,\bm}(\bx)| f_\tau\rd\mu \le C\xi^{-1}N^{-1} 
 \left \{ \int \Big [ \sum_{i=1}^N |x_i-\gamma_i|   \Big ]^2 f_\tau\rd\mu  \right \}^{1/2}  
   = C\xi^{-1} N^{-1/2-\fa} 
\label{exp}
\ee
using Assumption 1 \eqref{Qbound}.
We can also estimate
\begin{align}
    \int_{E-b}^{E+b} \frac { \rd E'} { 2 b}   \sum_{i\in J^+} Y_{i,\bm}(E',\bx)
  \le & \int_{E-b}^{E+b}\frac { \rd E'} { 2 b}   \sum_{i\in J^-} Y_{i,\bm}(E',\bx) + 
CN^{-1} |J^+\setminus J^-|  \non \\
   = & \int_\R \frac { \rd E'} { 2 b}    \sum_{i\in J^-}Y_{i,\bm}(E',\bx) 
 + CN^{-1}|J^+\setminus J^-|+
  \Xi^+_{J,\bm}(\bx) \label{fol} \\
 \le & \int_\R \frac { \rd E'} { 2 b}    \sum_{i\in J}Y_{i,\bm}(E',\bx)+
CN^{-1}|J^+\setminus J^-|+CN^{-1}|J\setminus J^-|+
  \Xi^+_{J,\bm}(\bx), \non
\end{align}
where the error term $\Xi^+_{J,\bm}$, defined by \eqref{fol}, 
comes from indices $i\in J^-$ such that $x_i \not \in [E-b, E+b]+O(1/N)$.
It satisfies the same bound \eqref{exp} as $\Omega^+_{J,\bm}$.
By the continuity  of $\varrho$, the density of $\gamma_i$'s is
bounded by $CN$, thus $|J^+\setminus J^-|\le CN\xi$ and
$|J\setminus J^-|\le CN\xi$.
 Therefore, 
summing up the formula \eqref{OO} for $i\in J$,
 we obtain from \eqref{uppe} and \eqref{fol}
$$
    \int_{E-b}^{E+b}\frac { \rd E'} { 2 b}   \int \sum_{i=1}^N
Y_{i,\bm}(E',\bx)  f_\tau\rd\mu\le \int \frac{1}{N}\sum_{i\in J}
 G \Big( N(x_i-x_{i+m_2}), \ldots \Big)
    f_\tau\rd\mu + C\xi + C\xi^{-1} N^{-1/2-\fa} 
$$
for each $\bm\in S_n$. 
A similar lower bound can be obtained analogously, and after choosing
$\xi= N^{-1/4}$, we obtain
\be
  \Bigg|  \int_{E-b}^{E+b} \frac { \rd E'} { 2 b}   \int \sum_{i=1}^N
 Y_{i,\bm}(E',\bx) f_\tau\rd\mu-  \int \frac{1}{N}
\sum_{i\in J} G \Big( N(x_i-x_{i+m_2}), \ldots \Big)
    f_\tau\rd\mu \Bigg|\le CN^{-1/4}
\label{seccc}
\ee
for each $\bm\in S_n$.

 Adding up \eqref{seccc}
for all $\bm\in S_n(M)$, we get
\be
\begin{split} 
  \Bigg|   \int_{E-b}^{E+b} \frac { \rd E'} { 2 b}  & \int\sum_{\bm\in S_n(M)}\sum_{i=1}^N
 Y_{i,\bm}(E',\bx) f_\tau\rd\mu   \\
 & -  \int \!\!\sum_{\bm\in S_n(M)} \frac{1}{N}
\sum_{i\in J} G \Big( N(x_i-x_{i+m_2}), \ldots \Big)
    f_\tau\rd\mu \Bigg|\le CM^{n-1}N^{-1/4},
\label{seccc1}
\end{split} 
\ee
and the same estimate holds for the equilibrium, i.e.,
if we set $\tau=\infty$ in \eqref{seccc1}.
 Subtracting these two formulas and applying \eqref{GG} from
 Theorem \ref{thmM}
to each summand on the second term in \eqref{seccc}  and using
\eqref{bbound2}, we conclude 
that
\be
    \Bigg|  \int_{E-b}^{E+b} \rd E' \int \sum_{\bm\in S_n(M)}\sum_{i=1}^N
 Y_{i,\bm}(E',\bx) (f_\tau\rd\mu -\rd \mu) \Bigg|\le  CM^{n-1}(N^{-1/4}+ N^{-\fa/2}).
\label{ssd}
\ee
Choosing 
\be
    M : =N^{\min \{ 1/4, \fa/2\}/n},
\label{Mdef}
\ee
we obtain that \eqref{ssd} vanishes as $N\to\infty$, and this proves \eqref{goal1}.

\bigskip
\noindent
{\it Step 2. Large $\bm$ case; proof of \eqref{goal2}.}
\bigskip

For a fixed $y\in \R$, $\ell >0$, let
$$
   \chi(y,\ell) : =\sum_{i=1}^N {\bf 1}\Big\{ x_i\in \big[ y- \frac{\ell}{N},
 y +\frac{\ell}{N}\big] \Big\}
$$
denote the number of points in the interval $[y-\ell/N, y+\ell/N]$.
 Note that for a fixed $\bm=(m_2, \ldots , m_n)$, we have
\be
   \sum_{i=1}^N |Y_{i,\bm} (E',\bx)| \le C\cdot\chi(E',\ell)
 \cdot {\bf 1}\Big(\chi(E',\ell)\ge m_n\Big) \le C\sum_{m=m_n}^\infty m \cdot
 {\bf 1}\Big(\chi(E',\ell)\ge m\Big),
\label{Ofi}
\ee
where $\ell$ denotes the maximum of $|u_1|+\ldots + |u_n|$
in the support of  $\wt O(u_1, \ldots , u_n)$.

 Since the summation over
all increasing sequences
$\bm = (m_2, \ldots, m_n)\in \N_+^{n-1}$ with a fixed $m_n$
contains at most $m_n^{n-2}$ terms,
we have
\be
  \sum_{\bm \in S_n^c(M)} \Bigg| \int_{E-b}^{E+b}  \rd E' \; 
\int \sum_{i=1}^N|Y_{i,\bm} (E',\bx)|  f_\tau \rd\mu \Bigg|
 \le C \int_{E-b}^{E+b}  \rd E' \;\sum_{m=M}^\infty m^{n-1}
  \int  {\bf 1}\Big(\chi(E',\ell)\ge m\Big) f_\tau\rd\mu.
\label{toc}
\ee
 Now we use  Assumption 3   for the interval
$I = [E' - N^{-1+\sigma}, E' + N^{-1+\sigma}]$ with 
$\sigma:=\frac{1}{2n}\min \{ 1/4, \fa/2\}$. 
 Clearly $\cN_I\ge \chi(E',\ell)$ for
sufficiently large $N$, thus we get from  \eqref{ass4} that
$$
   \sum_{m=M}^\infty m^{n-1}
  \int  {\bf 1}\Big(\chi(E',\ell)\ge m\Big) f_\tau\rd\mu 
 \le C_a \sum_{m=M}^\infty m^{n-1} \Big(\frac{m}{N^\sigma}\Big)^{-a} 
$$
holds for any $a\in \N$. By the choice of $\sigma$,
we get that $\sqrt{m}\ge N^\sigma$ for any $m\ge M$
(see \eqref{Mdef}), and thus choosing $a=2n+2$, we get
$$
   \sum_{m=M}^\infty m^{n-1}
  \int  {\bf 1}\Big(\chi(E',\ell)\ge m\Big) f_\tau\rd\mu 
 \le \frac{C_a}{M} \to 0
$$
as $N\to\infty$.
Inserting this into \eqref{toc},
this completes the proof of \eqref{goal2} and the
proof of Theorem \ref{thm:main}. \qed

\section{Local semicircle law and proof of  Lemma \ref{lm:bbound}}
\label{sec:best}

We first recall the  local semicircle 
law concerning 
the eigenvalues $x_1<x_2< \ldots < x_N$ of $H$.  Let
$$
  m(z) : = \frac{1}{N} {\rm Tr}\; \frac{1}{H-z} = \frac{1}{N}\sum_{j=1}^N
  \frac{1}{x_j-z}
$$
be the Stieltjes transform of the empirical eigenvalue distribution
at spectral parameter $z=E+i\eta$, $\eta>0$, and let
$$
  m_{sc}(z): = \int \frac{\varrho_{sc}(x)}{x-z}\rd x
$$
be the Stieltjes transform of the semicirle distribution.
In Theorem 4.1 of \cite{ERSY} we proved the following version
of the local semicircle law  (we remark that,
contrary to what is stated in Theorem 4.1 of \cite{ERSY},
  condition (2.5) of \cite{ERSY} is not needed and has not been
used in the proof): 

\begin{theorem} [Local semicircle law]\label{thm:semi} 
 Assume that the  distribution  $\nu$ of the matrix elements of the symmetric
Wigner matrix ensemble satisfies \eqref{logsob} with some  constant $\theta$ 
and assume that $y$ is such that $(\log N)^4/N \leq |y| \leq 1$. 
Then for any $K>0$ there exist positive constants $\delta_0$, $C$ and $c>0$, depending
only on $K$ and  $\theta$, such that
for any $|x|\le K$ we have
\be
\P \left( \left| m (x+iy) - m_{\text{sc}} (x+iy) \right| \geq \delta
\right) \leq 
C\, e^{-c \delta \sqrt{N|y| \, |2-|x||}}, 
\label{Pm}
\ee
for all $0\le\delta \le \delta_0$ and for all $N$ large enough.
\end{theorem}

As a corollary of the local semicircle law, the number of eigenvalues up to a fixed energy $E$ 
can be estimated. 
The precise  statement is the following proposition and it was proven in 
 Proposition 4.2, equation
(4.14)  of  \cite{ERSY}.

\begin{proposition}
  Assume that the  distribution  $\nu$ of the matrix elements of the symmetric
Wigner matrix ensemble satisfies \eqref{logsob} with some finite constant $\theta$. 
Let
\be
 n(E): = \frac{1}{N}\E\#\{ x_j\le E\}
\label{ndefi}
\ee
be the expectation of the empirical distribution function of the eigenvalues
and recall the definition of $n_{sc}(E)$ from \eqref{def:Nsc}.
Then there exists a constant $C >0$, depending only on
the constant $\theta$ in \eqref{logsob}, such that
\be
  \int_{-\infty}^\infty |n(E) - n_{sc} (E)| \rd E  \leq \frac{C}{N^{6/7}} \, .
\label{NNint}
\ee
\end{proposition}

The local semicircle law implies that the local density of eigenvalues  is bounded, 
but the estimate in  Theorem \ref{thm:semi}
deteriorates near the spectral edges. The following upper bound on the
number of eigenvalues in an interval provides a uniform control near the edges. 
This lemma was essentially proved in Theorem 4.6 of \cite{ESY3}
using ideas from an earlier version, Theorem 5.1 of \cite{ESY1}. 
For the convenience of the reader, a detailed proof is given in the Appendix \ref{C}.

\begin{lemma}\label{lm:upper} [Upper bound on the number of eigenvalues]
Consider a Wigner matrix with single entry distribution $\nu$ that satisfies
the logarithmic Sobolev inequality \eqref{logsob} with some constant $\theta$.
 Let  ${\cal N}_I$ denote
the number of eigenvalues in an interval $I\subset \R$.
Suppose that $|I|\ge (\log N)^2/N$, then there exist positive
constants $C$, $c$ and $K_0$, depending only on 
$\theta$, such that
\be\label{upbound}
 \P({\cal N}_I\ge KN|I|)\le Ce^{-c\sqrt{KN|I|}}
\ee
for any $K\ge K_0$.
\end{lemma}

We remark that for  Theorem \ref{thm:semi} and Lemma \ref{lm:upper}
it is sufficient to assume only the Gaussian decay condition \eqref{gauss}
instead of the logarithmic Sobolev inequality \eqref{logsob}.

\bigskip

We can now start to prove  Lemma \ref{lm:bbound}. 
Since the logarithmic Sobolev inequality  holds for $\nu$ \eqref{logsob}, it
also holds for  $\nu_t$  as well; for a proof
see  Lemma \ref{lm:lsiconv} in Appendix \ref{sec:lsi}
and recall that
 $\nu_t$  is the convolution of $\nu$ with the Ornstein-Uhlenbeck kernel which itself
satisfies the logarithmic Sobolev inequality.
 Moreover, the LSI constant of $\nu_t$
is bounded uniformly for all $t>0$, since it
is the maximum of the LSI constant $\theta$ of $\nu$ and
the LSI constant of the Ornstein-Uhlenbeck kernel, which is bounded uniformly in time.
Therefore Lemma \ref{lm:bbound}
follows immediately from its time independent version:

\begin{lemma}\label{lm:timeindep} Suppose that the distribution $\nu$ of the 
matrix elements of the symmetric
Wigner ensemble satisfies \eqref{logsob} with some
finite constant $\theta$. Then there exist positive
constants $\fa$, $\fb$  and $\fc$
such that
\be
   \E \Bigg[\frac{1}{N}\sum_{j=1}^N |x_j-\gamma_j|\Bigg]^2 \le N^{-1-2\fa}
\label{Qest}
\ee
and
\be 
    \P\Big\{ \max_{j=1,\ldots, N} |x_j-\gamma_j|\ge N^{-\fb} \Big\}
   \le \exp \big[ - N^\fc \big]
\label{xgamma1}
\ee
hold for the eigenvalues $x_j$ of $H$ and
for any $N\ge N_0= N_0(\theta, \fa,\fb,\fc)$.
\end{lemma}

The proof of Lemma \ref{lm:timeindep} is divided into two
steps.  In the first step,
Section \ref{sec:fluc},  we estimate the fluctuation of the eigenvalues
$x_j$ around their mean values
using the logarithmic Sobolev inequality. In the second step, Section
\ref{sec:dev},
we estimate the deviation of the mean location of $x_j$
from the classical location $\gamma_j$ using \eqref{NNint}.

\subsection{Fluctuation of the eigenvalues around their mean}\label{sec:fluc}

Denote by
\be
  \al_j  =  \E x_j
\label{def:al}
\ee
the expected location of $x_j$.
We start with an estimate
on the expected location of the extreme eigenvalues:

\begin{lemma}\label{lm:trunca}
Suppose that the probability measure $\nu$ of the matrix entries satisfies
\be
   \int_\R  e^{ c_0 |x| } \rd\nu(x) <\infty
\label{poldec}
\ee
for some $c_0> 0$ (this condition is satisfied, in particular, under 
\eqref{logsob}, see \eqref{gauss}).
 Then for any $\delta>0$ we have
\be
   -2-C_0N^{-1/4+\delta}\le \al_1 < \al_N \le 2+C_0N^{-1/4+\delta}
\label{a1}
\ee
with some constant $C_0$ depending on $\delta$ and $c_0$.
\end{lemma}

\begin{proof}
For any $M$, define the probability  measure  
\[
\zeta_M(\rd x) := Z_M^{-1} 1(|x| \le M)\nu(\rd x)
\]
on $\R$, where $Z_M$ is the normalization factor. 
Setting $M = N^{\delta}$ and using \eqref{poldec}, the total variational
norm between $\nu$ and $\zeta_M$ is bounded
\be
\| \zeta_M- \nu\|_\infty \le C e^{-cN^\delta}.
\label{zetadef}
\ee
 Denote by $\zeta_M^{ N}  = \bigotimes_{i\le j} \zeta_M$ ($\nu^{N}$ resp.) the 
probability law of the
random matrices whose matrix  elements are distributed according to $\zeta_M$ 
($\nu$ resp.). As usual, we neglect the fact that distribution $\nu$
should be replaced with $\wt\nu$ for the diagonal elements $i=j$. 
Since the number of index pairs $i\le j$ is of order $N^2$,
the total variational norm between $\zeta_M^{ N} $ and $\nu^N$
is bounded by
\be
\| \zeta_M^{ N} - \nu^{N}\|_\infty \le C N^2 e^{-cN^\delta}.
\label{TVnorm}
\ee
{F}rom Theorem 1.4 of \cite{Vu} we obtain that for any $\delta> 0$
\be
   x_N \le 2 + C M^{1/2} N^{-1/4} \log N \le 2 + N^{-1/4 + 
\delta}
\label{xnupper}
\ee
holds almost surely w.r.t. $\zeta_M^{ N}$. It follows from
\eqref{TVnorm}  that $x_N$ is bounded w.r.t. the distribution
$\nu^N$ as well, up 
to a subexponentially small probability. 
To estimate the tail of $x_N$ w.r.t. $\nu^N$,
we use that $\max_j |x_j|^2\le \mbox{Tr}\, H^2$ and 
the trivial large deviation bound based upon \eqref{poldec},
\be
  \P_{\nu^N}\big( \max_j |x_j|
\ge KN\big) \le \P_{\nu^N}\Big(\sum_{ij} |h_{ij}|^2 \ge (KN)^2 \Big)
  \le N^2 \int {\bf 1}(|y|^2\ge K^2N) \rd\nu(y)\le 
Ce^{-cK\sqrt{N}},
\label{xntail}
\ee
that holds for any $K>0$ with constants $C,c$ depending on $c_0$.
We thus obtain that
the expectations of $x_N$ w.r.t. these two measures satisfy
\[
\left | \E_{ \zeta_M^{ N}} x_N - \E_{\nu^N} x_N  \right |\le  C N^2 e^{-
cN^\delta}.
\]
{F}rom \eqref{xnupper} we also have
\[
\E_{ \zeta_M^{ N}} x_N \le 2 + N^{-1/4 + 
\delta}.
\]
Thus we have proved that, for any $\delta> 0$,
\[
\al_N=\E_{\nu^N} x_N  \le 2 + CN^{-1/4 + \delta}
\]
with some constant $C$ depending on $c_0$ and $\delta$.
Similar lower bound holds for $\al_1$.
\end{proof}

\medskip

Next we estimate the fluctuations of $x_j$:

\begin{proposition}\label{prop:bob} For any $\e>0$ we have
\be\label{3.1}
\P \Big( \max_j |x_j - \al_j | \ge N^{-1/2+\e}\Big) \le C e^{-cN^{\e}}.
\ee
with a constant $C$ depending on $\e$ and $\theta$.
\end{proposition}

\begin{proof} First order perturbation theory of the eigenvalue $x_j$ of $H$ shows that
$$
   |\nabla x_j|^2: = \sum_{\ell, k} \Big| \frac{\pt x_j}{\pt x_{\ell k}}\Big|^2
  = \frac{1}{N}\sum_{\ell,k} \Big| \frac{\pt x_j}{\pt h_{\ell k}}\Big|^2
  \le \frac{C}{N}\sum_{\ell,k} | \bu_j(k)\bu_j(\ell)|^2 = \frac{C}{N},
 $$
where $x_{\ell k}=N^{1/2} h_{\ell k}$ are the unscaled random variables,
see \eqref{scaling}, and $\bu_j=(\bu_j(1), \bu_j(2), \ldots , \bu_j(N))$
is the normalized eigenvector belonging to $x_j$.

 Using \eqref{logsob} and  the Bobkov-G\"otze concentration 
inequality (Theorem 2.1 of \cite{BG}), 
we have for any $T>0$
$$
\P\Big( \max_j |x_j - \E x_j | \ge \gamma\Big)
\le 2N \max_j \P  \Big( x_j - \E x_j  \ge \gamma\Big)
\le 2N e^{-\gamma T} \max_j \E \, e^{ \theta T^2 |\nabla x_j|^2}
\le CNe^{-c\gamma^2N}
$$
after optimizing for $T$ and using that $|\nabla x_j| \leq C N^{-1/2}$ from above.
%(see the analogous calculation in the proof of Theorem 3.1 of \cite{ESY1}).
This proves \eqref{3.1}.
\end{proof}

\medskip

The following proposition is a refinement of Proposition \ref{prop:bob}:

\begin{proposition}\label{prop:fluc} Fix a sufficiently small constant 
$\delta>0$  and
set $\kappa = N^{-1/18+\delta}$. Then
for any index $i$ with  $CN\kappa^{3/2}\le i\le N(1-C\kappa^{3/2})$  we have
\be
 \P \Big( |x_i-\al_i| \ge N^{-5/9+2\delta} \Big)\le C e^{-cN^{\delta}}
\label{xifluc}
\ee
and
\be
\P\Big( \frac{1}{N}\sum_{i=1}^N |x_i-\al_i| \ge N^{-5/9+2\delta}\Big)\le  C e^{-cN^{\delta}}.
\label{xal}
\ee
The constants $C$ and $c$  depend on $\delta$ and $\theta$ but are independent of $N$.
\end{proposition}

As a preparation for the proof of Proposition \ref{prop:fluc}, we need
the following estimate on the tail of the gap distribution.

\begin{lemma}\label{lm:gaptail}  Let  $|E|<2$. Denote by
$x_\al$ the largest eigenvalue below $E$ and assume
that $\al\leq N-1$. Then there exist positive constants $C$ and $c$,
depending only  on the Sobolev constant $\theta$ in \eqref{logsob},
such that for any $M$ that satisfy $c(\log N)^4/(2-|E|) \le M \le CN(2-|E|)$, we 
have
\be\label{gap}
   \P \Big(x_{\al+1} -E\ge \frac{M}{N}, \; \al\leq N-1\Big)
\leq C\; e^{-c  [2-|E|]^{3/2}  \sqrt{M}}. 
\ee
\end{lemma}

This lemma was proven in Theorem E1 of  \cite{ERSY}, see also  Theorem 3.3 of \cite{ESY3},
and  the proof will not be repeated here. We only mention 
the main idea,  that the local semicircle law,  Theorem \ref{thm:semi},
provides a positive {\it lower bound} on the number of eigenvalues 
in any interval $I$ of size $|I|\ge A(\log N)^4/\big[N\big|2-|x|\big|\big]$
around the point $x$ with a sufficiently large constant $A$.  
In particular, it follows that there is at least one eigenvalue
in each  such interval $I$ with a very high probability. \qed

\bigskip

{\bf Proof of Proposition \ref{prop:fluc}. }
 We  choose $M, K$ positive
numbers, depending on $N$, such that
% for small $\delta>0$
\be
%N^{-1/6+2\delta}\le \kappa \le N^{-\delta}, \qquad
 M\kappa^6\ge N^{2\delta},
\qquad K \le cN\kappa^{3/2}
\quad\mbox{and}\quad c(\log N)^4/\kappa \le M \le CN\kappa
\label{KM}
\ee
with some sufficiently small $c>0$ and large $C>0$ constants.
Let
\be
  \Phi: = 2N^{-1/2+2\delta} K^{-1/2} + \frac{2KM}{N}.
\label{Phidef}
\ee

Consider an index $i$ with
$  CN\kappa^{3/2} \le
i\le N(1-C\kappa^{3/2})$, then $|2-|\gamma_i||\ge C\kappa$.
We first show that  $|2-|x_i||\ge C\kappa^{2}$
with a very high probability. Suppose, in the contrary, that
$x_i< -2+C\kappa^{2}$
for some $i\ge CN\kappa^{3/2}$
(the case $x_i \ge 2- C\kappa^{2}$ is treated analogously).
{F}rom \eqref{a1} and \eqref{3.1} it follows that $x_1\ge -2-C_0N^{-1/4+\delta}$
with a very high probability. But then the interval
$[-2-C_0N^{-1/4+\delta}, -2+C\kappa^{2}]$ of length $ C_0N^{-1/4+\delta}
+ C\kappa^{2}\ll \kappa^{3/2}$ would contain
$CN\kappa^{3/2}$ eigenvalues,  an event with an extremely low probability by 
\eqref{upbound}.

Knowing that $|2-|x_i||\ge \kappa^{2}$
with a very high probability,
we can use \eqref{gap} to conclude that
for any index $i$ with
$  CN\kappa^{3/2} \le
i\le N(1-C\kappa^{3/2})$ we have
\be\label{expgap}
\P\Big( x_{i+1}- x_i  \ge \frac{M}{N}\Big) \le e^{ - c\sqrt{M\kappa^6}}
+ Ce^{-cN^{\delta}} \le  Ce^{-cN^{\delta}}
\ee
by \eqref{KM}.
Then
\be
\P \left ( \Big|x_i - \frac{1}{2 K + 1} \sum_{|j-i| \le K} x_j \Big|
\ge  \frac{KM}{N}
\right ) \le Ce^{ - cN^{\delta}}.
\label{xi}
\ee
Similarly to the calculation in Theorem 3.1 of \cite{ESY1},
by using the logarithmic Sobolev inequality \eqref{logsob}, we have
\be\label{8.3}
\P\Bigg(
\Big| \frac{1}{2 K + 1} \sum_{j\; : \; |j-i| \le K} x_j -
\frac{1}{2 K + 1} \sum_{j\; : \; |j-i| \le K} \E x_j \Big| \ge
N^{-1/2+2\delta}K^{-1/2}\Bigg) \le Ce^{ - cN^{\delta}}.
\ee
This bound holds for any index $i$ with the remark that
if $i< K$ or $i> N-K$, then the averaging over the indices $j$
is done asymmetrically.

Combining this estimate with \eqref{xi} we have, apart from
a set of very small probability, that
\be\label{7.3}
\Big| x_i- \frac{1}{2 K + 1} \sum_{j \; : \; |j-i| \le K} \al_j\Big| \le
\frac{\Phi}{2}
\ee
for any  $  CN\kappa^{3/2} \le
i\le N(1-C\kappa^{3/2})$.
Taking expectation, and using the tail estimate \eqref{xntail}
to control $x_N$ on the event of very small probability where
\eqref{7.3} may not hold,
 we also obtain for these $i$ indices that
\be\label{aa}
\Big| \al_i- \frac{1}{2 K + 1} \sum_{j \; : \;|j-i| \le K} \al_j\Big| \le
\frac{\Phi}{2}.
\ee
Subtracting the last two inequalities yields  
$$
 \P \Big( |x_i-\al_i| \ge \Phi \Big)\le C e^{-cN^{\delta}}
$$
and combining this bound with the estimate \eqref{3.1}
for the extreme indices, we obtain
$$
\P\Big( \frac{1}{N}\sum_{i=1}^N |x_i-\al_i| \ge C\kappa^{3/2}N^{-1/2+\delta} +
 \Phi \Big)\le  C e^{-cN^{\delta}}.
$$
The inequalities \eqref{xifluc}
and \eqref{xal} now follow if we
 choose the parameters as 
\be
  \kappa = N^{-1/18+\delta}, \quad M=N^{1/3}, \qquad K= N^{1/9},
\label{choice}
\ee
which choice is compatible with the conditions \eqref{KM}.
\qed

\subsection{Deviation of the eigenvalues from their classical 
locations}\label{sec:dev}

The next Proposition \ref{prop:xg} below estimates the distance of
the eigenvalues from their
location given by the semicircle law. This will justify
that the convex extension of the potential $W_j$ affects only
regimes of very small probability.

\begin{proposition}\label{prop:xg}
For any small $\delta>0$ and for any $j=1,2, \ldots N$  we have
\be
  \P \Big( |x_j-\gamma_j| \ge CN^{-1/5+\delta}\Big) \le Ce^{-cN^{\delta}}
\label{xgamma}
\ee
and we also have
\be
    \max_m  |\al_m -\gamma_m |\le C N^{-1/5+\delta}.
\label{ag}
\ee
The constants $C$ and $c$ depend on $\delta$ and $\theta$ but are independent of $N$.
\end{proposition}

We remark that in the bulk $\al_m-\gamma_m$ is expected to
be   bounded by $O(N^{-1+\e})$ (in the hermitian case
it was proven in \cite{Gu}, see also \cite{TV});
near the edges one  expects   $\al_m-\gamma_m \sim O(N^{-2/3})$.  Our estimate
is not optimal, but it gives a short proof that is sufficient for our purpose.
 We remark that after submitting this paper, these conjectures
were proven in \cite{EYY3}.

\bigskip

{\bf Proof of Proposition \ref{prop:xg}.}
We define
$$
 \wt n(E): = \frac{1}{N} \#\{ \al_j\le E\}
$$
to be the counting function of the expected locations of the eigenvalues.
We compare $\wt n(E)$ with $n(E)$ defined in \eqref{ndefi}.
Using the fluctuation
bound \eqref{3.1},  we have
$$
 n\big( E- N^{-1/2+\e}\big) - Ce^{-cN^{\e}}  \le \wt n(E) \le n\big(E+N^{-
1/2+\e}
\big) + Ce^{-cN^{\e}}
$$
for any $E\in \bR$. In fact, the upper bound on the density \eqref{upbound} 
guarantees
that $\wt n(E)$ and $n(E)$ are Lipschitz continuous on any scale much bigger 
than $(\log N)^2/N$,
i.e.
\be
|\wt n(E)-n(E)|\le CN^{-1/2+\e} +Ce^{-cN^{\e}} \le CN^{-1/2+\e}
\label{compp}
\ee
for any $E$.

We  write
$$
\sum_{j=1}^N |\al_j-\gamma_j| =  \sum_{j: \al_j\ge \gamma_j} (\al_j
-\gamma_j)
+  \sum_{j: \al_j\le \gamma_j} (\gamma_j-\al_j).
$$
For the first term
\be
\begin{split}
  \sum_{j: \al_j\ge \gamma_j} (\al_j-\gamma_j)
  = & \int \rd E \sum_{j: \al_j\ge \gamma_j} {\bf 1}( \gamma_j
\le E \le \al_j)
  =  \int \rd E  \sum_{j: \al_j\ge \gamma_j} {\bf 1}\Big( \wt n(E) \ge 
\frac{j}{N}
 >  n_{sc}(E)\Big) \\
  = & N\int \rd E \; {\bf 1}(n_{sc}(E)\le\wt n(E))
(n_{sc}(E)-\wt n(E))
\end{split}
\ee
and the second term is analogous.
We thus  have
\be\label{algamma1}
  \frac{1}{N}\sum_{j=1}^N |\al_j-\gamma_j| = \int |\wt n(E)-n_{sc}(E)|\rd E.
\ee
By Lemma \ref{lm:trunca}
\[
\int_{|E|\ge 3} |\wt n(E)-n_{sc}(E)|\rd E =0 .
\]
For the energy range $|E|\le 3$, we use \eqref{compp}:
\be
\begin{split}
\int_{|E|\le 3} |\wt n(E)-n_{sc}(E)|\rd E  \le & \; CN^{-1/2+\e}
+ \int_{|E|\le 3} | n(E)-n_{sc}(E)|\rd E\\
\le &  CN^{-1/2+\e}
\end{split}
\ee
from \eqref{NNint}. Thus we obtain from \eqref{algamma1}
\be\label{algamma}
  \frac{1}{N}\sum_{j=1}^N |\al_j-\gamma_j| \le   C N^{-1/2+\e}.
\ee
with an $\e$ dependent constant.

\medskip

To estimate $|\al_m-\gamma_m|$, we can assume,
without loss of generality,
that $\al_m \ge \gamma_m$, the other case is treated
analogously. Let $\lambda>0$ be a parameter that will be
optimized later.
Set $m_0=[CN\lambda^{3/2}]$ with a sufficiently large constant $C$.
Since $n_{sc} (-2+ \delta)\sim \delta^{3/2}$, for any small $\delta>0$,
the parameter $\lambda$
is  roughly the energy difference from the edge to the $m_0$-th eigenvalue.

First we consider  an index $m$ such that $m_0 \le m \le N- m_0$.
For a small positive number $\ell$, define
$$
S:= \{ j\; : \; \gamma_j \in [\gamma_m , \gamma_m+\ell]\}.
$$
{F}rom the property $n_{sc} (-2+ \delta)\sim \delta^{3/2}$ for small $\delta$,    
we have
$$
  |S|  \ge cN\ell \lambda^{1/2}.
$$
Now set $\ell = \min \{ \frac{1}{2}|\alpha_m-\gamma_m|, c\lambda\}$
with some small positive constant $c$.
Since for all $j\in S$
$$
  \al_j -\gamma_j \ge \al_m - (\gamma_m+\ell) \ge \frac{1}{2}(\al_m-
\gamma_m) \ge \ell,
$$
we have
$$
 \sum_{j=1}^N |\al_j-\gamma_j| \ge \sum_{j\in S} \ell
  \ge  cN\ell^2\lambda^{1/2}.
$$
Combining this estimate with \eqref{algamma},
we have
$$
 \ell \le C\lambda^{-1/4}N^{-1/4+\e/2}.
$$
Assuming that $\lambda\ge CN^{-1/5+2\e/5}$, we see that
$\ell =  \frac{1}{2}|\alpha_m-\gamma_m|$ and we obtain
\be\label{8.8}
   |\alpha_m-\gamma_m| \le C N^{-1/5+\e}
\ee
for any $m$ with $m_0 \le m \le N- m_0$.

For the extreme indices, we use that
if $m\le m_0$, then from Lemma  \ref{lm:trunca} and \eqref{8.8}, we have
$$
   -2- C N^{-1/4+\delta}
  \le \al_1\le \al_m \le \al_{m_0} \le \gamma_{m_0} + C N^{-1/5+\e}
   \le -2 + C\lambda + C N^{-1/5+\e}
$$
and
$$
   -2 \le \gamma_m \le -2 +C\lambda
$$
for all $m\le m_0$. Thus
$$
   |\al_m - \gamma_m|\le  C\lambda + C N^{-1/5+\e}
$$
with $C$ depending on $\e$.
Similar estimates hold at the upper
edge of the spectrum, i.e. for $m\ge N-m_0$.
Choosing $\lambda = CN^{-1/5+\e}$,
we conclude the proof  of \eqref{ag}. The proof of
\eqref{xgamma} then follows from \eqref{3.1} and this
concludes the proof
of Proposition \ref{prop:xg}. \qed

\bigskip
The following Proposition is a strengthening of the bound \eqref{algamma}
used previously.

\medskip

\begin{proposition}\label{prop:algam} With the previous notations, we have
\be
 \frac{1}{N} \sum_i |\al_i-\gamma_i|\le CN^{-5/9+2\delta}.
% C\Phi+C\kappa N^{-1/2+\delta}.
\label{algg}
\ee
for any $\delta>0$ and with a constant $C$ depending on $\delta$ and $\theta$.
\end{proposition}

\begin{proof} Recalling the definition of $\Phi$ from \eqref{Phidef}, we
will prove that 
\be
\frac{1}{N} \sum_i |\al_i-\gamma_i|\le
C\Phi+C\kappa N^{-1/2+\delta},
\label{algg1}
\ee
which gives \eqref{algg} with the choice of parameters \eqref{choice}.
We proceed similarly to the proof of Proposition \ref{prop:xg}
but we notice that in addition to \eqref{compp}, a stronger
bound on $|n(E)-\wt n(E)|$ is available  for $E\in I:=[E_-, E_+]$, where
$E_\pm:= \pm(2-C_2\kappa)$, with some large constant $C_2$ and setting
 $\kappa:= N^{-1/18+\delta}$ as in Proposition \ref{prop:fluc}.
 To obtain
an improved bound, note
that for any $E$ in this interval
\be
\begin{split}\label{ke}
  n(E) = & \frac{1}{N}\sum_j \E \; {\bf 1}(  x_j \le E) \le
 \frac{1}{N}\sum_j \E \; {\bf 1}(  \al_j \le E +\Phi)  +  Ce^{-cN^\delta}
 =\wt n(E+\Phi) +  Ce^{-cN^\delta}.
\end{split}
\ee
To see this inequality,   define the random index
$$
j_0= j_0(E):=\max\{ j \; : \; x_j\le E\} = \sum_j \; {\bf 1}(  x_j \le E)
$$
and the deterministic index
$$
 j_1=j_1(E) := \max\{ j \; : \; \al_j\le E+\Phi\} = \sum_j \; {\bf 1}( \al_j \le 
E+\Phi).
$$
The estimate \eqref{ke} will then follow if we prove that $j_0 \le j_1$,
i.e. $\al_{j_0}\le E+\Phi$,
with a very high probability.
By \eqref{xgamma} we have, with a very high probability, that
$$
   \gamma_{j_0} - CN^{-1/5+\delta} \le x_{j_0}\le E\le
 x_{j_0+1}\le \gamma_{j_0+1} +CN^{-1/5+\delta} \le \gamma_{j_0} +CN^{-
1/5+\delta}.
$$
Therefore, with a very high probability, $\gamma_{j_0}$ is in
the $CN^{-1/5+\delta}$ vicinity of $E\in I$,
and thus $CN\kappa^{3/2} \le j_0 \le N(1-C\kappa^{3/2})$ holds
for any fixed $C$ if $C_2$ in the definition of $E_\pm$ is sufficiently large.
Thus $|x_{j_0}-\al_{j_0}|\le \Phi$ with  a very high probability by 
\eqref{xifluc},
so $x_{j_0}\le E$ implies $\al_{j_0}\le E+\Phi$ and this proves
\eqref{ke}.

The proof of the lower bound
$$
n(E)\ge \wt n(E-\Phi) -  Ce^{-cN^\delta}
$$
is analogous.
Finally, by the Lipschitz continuity of $\wt n(E)$  on a scale bigger than 
$(\log N)^2/N$,
we have
\be
\begin{split}\label{ke1}
  |n(E) -\wt n(E)| \le C\Phi,\qquad \forall E\in [E_-, E_+],
\end{split}
\ee
where we also used that $\Phi\ge C\exp (-cN^\delta)$.

Define the interval $J=[-2-C_1N^{-1/4+\delta}, 2+ C_1N^{1/4+\delta}]$
with a constant $C_1$  larger than the constant $C_0$ in \eqref{a1}.
Using \eqref{algamma1}, we have
\be\label{ins}
  \frac{1}{N}\sum_{j=1}^N |\al_j-\gamma_j| \le (I)+(II)+(III)
 + \int |n(E)-n_{sc}(E)|\rd E
\ee
with
$$
  (I): =\int_{I} |\wt n(E)-n(E)|\rd E, \qquad
(II): = \int_{J\setminus I} |\wt n(E)-n(E)|\rd E, \qquad
(III): =\int_{J^c} |\wt n(E)-n(E)|\rd E.
$$
{F}rom \eqref{ke1} and \eqref{compp}, we have
$$
  (I) \le C\Phi, \qquad (II)\le CN^{-1/2+\delta}|J\setminus I| \le C\kappa N^{-
1/2+\delta}
$$
since $|J\setminus I|\le C\kappa + N^{-1/4+\delta}\le C\kappa$.
Finally, $\wt n(E) \equiv 0$ for $E<-2-C_0N^{-1/4+\delta}$ and $\wt n(E) \equiv 1$
for $E>2+C_0N^{-1/4+\delta}$  by \eqref{a1}.  
Since $C_1>C_0$,
combining these estimates with the
fluctuation \eqref{3.1} and with the tail estimate \eqref{xntail},
we obtain that  $n(E)(1-n(E))\le C\exp\big[ - cN^{1/4}\big]$
 for any $E\in J^c$, and it decays exponentially for large $|E|$,
 therefore
$$
  (III)\le Ce^{-cN^{1/4}}.
$$
Collecting all these estimates, inserting them into \eqref{ins} and using
\eqref{NNint}, we obtain \eqref{algg1} and  conclude the proof of Proposition \ref{prop:algam}.
\end{proof}

\bigskip

Finally, we can complete the 
proof of Lemma \ref{lm:timeindep}. 
By \eqref{xal} and \eqref{algg},
we have
\[
 \frac 1 N \sum_{ k } |x_k - \gamma_k|
\le    \frac 1 N \sum_{k}
\Big[|\alpha_k - \gamma_k|+ |x_k- \alpha_k |\Big]
\le   C N^{-5/9+2\delta}
\]
apart from a set of probability $C\exp\big[-cN^\delta\big]$.
Combining it with the tail estimate  \eqref{xntail} 
on $\max|x_k|$, we obtain \eqref{Qest} with any $\fa < 1/18$. 
The inequality \eqref{xgamma1} in  Lemma \ref{lm:timeindep}
follows immediately from \eqref{xgamma}
with any $\fc>0$ sufficiently small and 
with $\fb < 1/5-\fc$. \qed

\appendix

\section{Some Properties of the Eigenvalue Process}

In the main part of the paper we did not specify the function
spaces in which the equations \eqref{dy} and \eqref{dytilde} are solved.
In this appendix we summarize some basic properties of these equations.
In particular, we justify the integration by parts in \eqref{1.5}.
For simplicity, we consider the most singular $\beta=1$ case only.

The Dyson Brownian motion as a stochastic process
was rigorously constructed in Section 4.3.1 of
\cite{G}. It was proved that   the eigenvalues do not collide
with probability one and thus  \eqref{dy} holds in a weak sense
on the open set $\Sigma_N$.
The coefficients of $L$ have a
$(x_i-x_{j})^{-1}$ singularity near the coalescence hyperspace
$x_i=x_{j}$. We focus only on the single collision singularities,
i.e. on the case $j=i\pm 1$. By the ordering of the eigenvalues,
higher order collision points form a zero measure set
on the boundary of $\Sigma_N$ and can thus be neglected.
In an open neighborhood near the coalescence hyperspace $x_i=x_{i+1}$, the 
generator
has the form
$$
   L = \frac{1}{2N} \Big( \pt_i^2+\pt_{i+1}^2 +
 \frac{1}{x_{i+1}-x_{i}}(\pt_{i+1} -\pt_{i}) \Big)+L_{reg}
 =  \frac{1}{4N} \Big(\pt_v^2+\pt_u^2 +
 \frac{1}{u}\pt_u \Big)+ L_{reg}, \qquad u>0,
$$
after a change of variables, $v=\frac{1}{2}(x_i+x_{i+1})$,
$u=\frac{1}{2}(x_{i+1}-x_{i})$, where $L_{reg}$ has regular
coefficients. The boundary condition at $u=0$ is given by
the standard boundary condition of the generator of the
Bessel process, $\pt_u^2 +\frac{1}{u}\pt_u$,
which is $uf'(u)\to 0$ as $u\to 0+$. Thus $L$ is defined on
functions $f\in C^2(\Sigma_N)$ with sufficient
decay at infinity and with boundary conditions
\be
  \lim_{x_{i+1}-x_i\to 0} (x_{i+1}-x_i)(\pt_{i+1}-\pt_i)f
   \to 0
\label{bc}
\ee
for each $i$.

The generator $\wt L$ of \eqref{dytilde}
differs from $L$ only in drift terms with  bounded coefficients,
hence the boundary conditions of $\wt L$ and $L$ coincide.
Finally, we need some non-vanishing and regularity property of
the solution of  \eqref{dytilde}:

\begin{lemma}  Let $\Omega \subset \Sigma^N$ be a bounded open set
 such that
$$
  \overline\Omega \cap \bigcup_{i<j\; : \; (i,j)\ne (1,2)}
\{ \bx \; : \; x_i=x_j\}
 =\emptyset
$$
i.e. $\overline\Omega$ intersects
at most one of the coalescent hyperplanes, namely the
$\{ x_1=x_2\}$. Then any weak solution  $q_t(\bx)$ of \eqref{dytilde}
with boundary conditions \eqref{bc}
is $C^2$ on $(t, \bx)\in R_+\times\overline\Omega$ and for any $t>0$ we have
$$
    0 < \inf_{\overline\Omega} q_t \le  \sup_{\overline\Omega} q_t <\infty
$$
\end{lemma}

\begin{proof}
The statement follows from regularity properties
of the Bessel process with generator $\pt_u^2+\frac{1}{u}\pt_u$.
In a small neighborhood of the coalescence line
$x_1=x_2$ one can introduce a local coordinate system $(u, \by)=\Phi(\bx)$, 
where
$u=\frac{1}{2}(x_2-x_1)>0$, $\by\in \bR^{N-1}$, so that, in the case for GOE,
$$
  \wt L = \frac{1}{4N} \Big[ \pt_u^2 + \frac{1}{u} \pt_u \Big] + L_{reg},
$$
where $L_{reg}$ is an elliptic operator with second derivatives in
the $\by$ variables and with bounded  coefficients
on the compact set $\Phi(\overline\Omega)$.
The solution in the new coordinates is $\wt q_t(u, \by) =
q_t ( \Phi^{-1}(u, \by))$.
Introducing a function
$\widehat q_t(a, b, \by) := \wt q_t ( \sqrt{a^2+b^2}, \by)$ defined in
$N+1$ variables, we see
that $\wh q_t$ satisfies $\pt_t \wh q_t = \wh L \wh q_t$, where
$$
  \wh L = \frac{1}{N}\Big[ \partial_a^2 + \pt_b^2 \Big] +  L_{reg}.
$$
i.e. $\wh L$ is elliptic with bounded coefficients in the new variables.
Notice that the boundary condition \eqref{bc} implies that,
in the two dimensional plane
of $(a, b)$, the support of the test function for the equation
$\pt_t \wh q_t = \wh L \wh q_t$
is allowed to include the origin $(0,0)$.

By standard parabolic regularity, we obtain that  the solution is $C^2$ and
is bounded from above and below.

\end{proof}

\bigskip

This lemma justifies the integration by parts in \eqref{1.5}. Since
$q\in C^2$ and it is separated away from zero, $h=\sqrt{q}$ has no singularity
on the coalescence lines. Since the  function $\exp (-\wt \cH)$ vanishes 
whenever
$x_i=x_j$ for some $i\neq j$, the boundary terms of the form
$$
  \int_{x_i=x_j} \pt \sqrt{q}\;\;  \pt^2\!\sqrt{q} \; e^{-\wt \cH} \rd\bx\;
$$
in the integration by parts vanish.

\section{Logarithmic Sobolev inequality for convolution measures}\label{sec:lsi}

\begin{lemma}\label{lm:lsiconv} Suppose $K$ and $H$ are two probability 
densities
on $\R$  so that the logarithmic Sobolev inequality holds
with constants $a$ and $b$, respectively. Then 
logarithmic Sobolev inequality holds for their convolution
$K\ast H$
as
\be\label{A1}
\int_\R f(x) \log f(x) K\ast H(x) \rd x \le  \max (a, b)
\int_\R \big(\nabla \sqrt {f(x)}\big)^2 K\ast H(x) \rd x
\ee
for any $f$ with $\int f(x)  K\ast H(x) \rd x=1$. Here $\nabla=\rd/\rd x$.
\end{lemma}

{\bf Proof.} The following proof is really a special case
of the martingale approach used in \cite{LuY} to prove LSI.  Let
\[
g(y) = \int_\R f(x)K(x-y) \rd x.
\]
Then the left side of \eqref{A1} is equal to
\be
\int  \int_\R  \Big[ f(x) \log [f(x)/g(y) ] +   f(x) \log g(y) \Big] K(x-y) \rd 
x
 H(y)  \rd y.
\label{mart}
\ee
For any fixed $y$, from the LSI w.r.t. the measure $ K(x-y) \rd x$,
the first term on the right hand side is bounded by
\[
a \int_\R \int_\R ( \nabla \sqrt {f(x)} )^2  K(x-y) \rd x  H(y)  \rd y.
\]
Since $\int_\R g H=1$, the second term in \eqref{mart} is estimated by
\[
\int_\R g(y) \log g(y) H(y) \rd y \le b \int_\R (\nabla \sqrt {g(y)} )^2 H(y) 
\rd y
= \frac b 4  \int_\R  g(y)^{-1}  \left ( \int_\R f(x)\nabla_y K(x-y) \rd x
\right  )^2 H(y) \rd y.
\]
Integrating by parts, we can rewrite the last term as
\[
\frac b 4  \int_\R  g(y)^{-1}  \left ( \int f'(x) K(x-y) \rd x \right )^2 H(y) 
\rd y
\le  b  \int_\R   \int_\R (\nabla \sqrt {f(x)} )^2 K(x-y) \rd x H(y) \rd y,
\]
where we have used $f'(x) = 2 \sqrt {f(x)} \nabla \sqrt {f(x)}$
and the Schwarz inequality.
Combining these inequalities, we have proved the Lemma.
\qed

\section{Proof of Lemma \ref{lm:upper}}  \label{C}

\medskip

 For any $1\le k\le N$, let $H^{(k)}$ denote
the $(N-1)\times (N-1)$ minor that is obtained from the
Wigner matrix $H$ by removing the $k$-th row and column.
Let $\ba^{(k)} = (h_{k1}, h_{k2}, \ldots h_{k,k-1}, h_{k,k+1}, \ldots
h_{kN})^t$ be the $k$-th column of $H$ without the $h_{kk}$ element.
Let $\lambda_1^{(k)} < \lambda_2^{(k)}< \ldots <\lambda_{N-1}^{(k)}$ be
the eigenvalues and $\bu_1^{(k)}, \bu_2^{(k)}, \ldots$
the corresponding eigenvectors of $H^{(k)}$ and
set 
$$
 \xi_\al^{(k)}:= N|\ba^{(k)}\cdot \bu_\al^{(k)}|^2, \qquad \al=1,2,\ldots, N-1.
$$
It is well known (see, e.g. Lemma 2.5 of \cite{ESY1}), that
the eigenvalues of $H^{(k)}$ and $H$ are interlaced for each $k$, i.e. 
\be\label{interlace}
   x_1< \la_1^{(k)} < x_2 < \la_2^{(k)} < \ldots < x_{N-1} < \la_{N-1}^{(k)} < x_N.
\ee
Expressing the resolvent $G=(H-z)^{-1}$ of $H$ at a spectral parameter $z=E+i\eta$,
$\eta>0$, in terms of the resolvent of $H^{(k)}$, we
obtain
\be
   G_z(k,k)= \frac{1}{h_{kk} -z -  \ba^{(k)} \cdot (H^{(k)}-z)^{-1}\ba^{(k)}}
 =\Bigg[ h_{kk}-z - \frac{1}{N}
\sum_{\al=1}^{N-1}\frac{\xi_\al^{(k)}}{\lambda_\al^{(k)}- z}\Bigg]^{-1}.
\label{G11}
\ee
By considering only the imaginary part, we obtain
\be
  |G_z(k,k)|  \leq \eta^{-1} \Bigg| 1+ \frac{1}{N}
\sum_{\al=1}^{N-1}\frac{\xi_\al^{(k)}}{(\lambda_\al^{(k)}-E)^2
+\eta^2 }\Bigg|^{-1}.
\label{gz}
\ee

For the interval $I\subset\bR$ given in Lemma \ref{lm:upper},
set $E$ to be its midpoint and $\eta:=|I|$, i.e. $I= [E-\frac{\eta}{2},
E+\frac{\eta}{2}]$. Clearly
$$
  \cN_I \le C\sum_{j=1}^N \frac{\eta^2}{(x_j -E)^2+\eta^2} = \frac{C\eta}{\pi}
  \im \mbox{Tr}\, G(z),
$$
thus from \eqref{gz}  we obtain
\be
  \cN_I\leq C\eta\sum_{k=1}^N |G_z(k,k)| \leq
 CN\eta^2\sum_{k=1}^N \Big|
\sum_{\al: \lambda_\al^{(k)}\in I}\xi_\al^{(k)}\Big|^{-1},
\label{nxi}
\ee
where we restricted the $\al$ summation in \eqref{gz}
only to eigenvalues
lying in $I$.

For each $k=1,2,\ldots N$, we define the event
$$
  \Omega_k :=\Big\{ \sum_{\al: \lambda_\al^{(k)}\in I}\xi_\al^{(k)}
  \leq \delta (\cN_I-1)\Big\}
$$
for some small $\delta>0$.
By the interlacing property of the $\mu_\al$ and $\lambda_\al^{(k)}$
eigenvalues, we know that there
are at least $\cN_I-1$ eigenvalues of $H^{(k)}$ in $I$.
Since the logarithmic Sobolev inequality \eqref{logsob} implies that
the tail of the distribution $\nu$  has a Gaussian bound \cite{Le},
we can apply  Lemma \ref{lm:bour} below to conclude
that there exists a positive universal constant $c$
such that $\P(\Omega_k)\leq \E \exp{\big[-c\sqrt{\cN_I-1}\big]}$.
Setting $\wt\Omega
= \bigcup_{k=1}^N\Omega_k$, we see that
\be
\P(\wt\Omega \; \text{and} \; \cN_I\ge KN|I|)
\leq N \, \E \Big[ e^{-c\sqrt{\cN_I-1}}\cdot {\bf 1}\big( \cN_I\ge KN|I|\big)\Big] 
\leq e^{-c'\sqrt{KN|I|}}
\label{wt}
\ee
if $K$ is sufficiently large, recalling that $\eta=|I|\ge (\log N)^2/N$.
On the complement event, $\wt\Omega^c$, we have from \eqref{nxi} that
$$
  \cN_I\leq \frac{CN^2\eta^2}{ \delta(\cN_I-1)}
$$
i.e. $\cN_I\leq (C/\delta)^{1/2} N\eta$. Choosing $K$
sufficiently large, we obtain \eqref{upbound} from \eqref{wt}.
 This proves Lemma \ref{lm:upper}.

\medskip

\begin{lemma}\label{lm:bour}
Let the components of the vector $\bb\in \R^{N-1}$
be real i.i.d. variables with a common distribution $\rd \nu$
that satisfies a Gaussian decay condition \eqref{gauss}
for some positive $\delta_0>0$.
Let $\xi_\al = |\bb \cdot \bv_\al|^2$, where $\{ \bv_\al\}_{\al\in
{\cal I}}$
is an orthonormal set in $\R^{N-1}$. Then for
$\delta\leq 1/2$ there is a constant $c>0$
such that
\be
  \P\big\{ \sum_{\al\in {\cal I}}\xi_\al \leq \delta m\big\} \leq
e^{-c\sqrt{m}} \;
\label{lm:xii}
\ee
holds for any ${\cal I}$,
where $m=|{\cal I}|$ is the cardinality of the index set ${\cal I}$.
\end{lemma}

\bigskip

{\it Proof of Lemma \ref{lm:bour}.}  We will need the following
result of Hanson and Wright \cite{HW},
extended to non-symmetric variables by Wright \cite{Wr}.
We remark that this statement can also be extended
to complex random variables (Proposition 4.5 \cite{ESY3}).

\begin{proposition}\label{prop:HW}\cite{HW, Wr}
Let $b_j$, $j=1,2,\ldots N$ be a sequence of real i.i.d. random
variables with distribution $\rd\nu$
satisfying the Gaussian decay \eqref{gauss} for some $\delta_0>0$.
 Let $a_{jk}$, $j,k=1,2,\ldots N$ be
arbitrary real numbers and let $\cA$ be the $N\times N$ matrix
with entries $\cA_{jk}:= |a_{jk}|$.   Define
$$
 X:=\sum_{j,k=1}^N a_{jk} \big[ b_j {b}_k -\E  b_j {b}_k\big]\; .
$$
Then  there exists a constant $c>0$, depending only on $\delta_0, D$
from \eqref{gauss},
such that for any $\delta>0$
$$
\P (|X|\ge \delta)\leq 4\exp\big( -c\min\{\delta/A,\; \delta^2/A^2\}\big)\;,
$$
where $A:= (\text{Tr}\, \cA\cA^t)^{1/2}=\big[\sum_{j,k} |a_{jk}|^2\big]^{1/2}$. 
\end{proposition}

We will apply this  result for
$$
 X= \sum_{i, j=1}^N a_{ij} \big[ b_i{b}_j - \E \, b_i b_j\big],
\qquad \mbox{with}\qquad
  a_{ij}: = \sum_{\al\in {\cal I}} {v}_\al(i)v_\al(j)\, .
$$
Notice that  $\sum_{\al\in {\cal I} }\xi_\al = X + | {\cal I}|=X+m$
since  $\E \,\xi_\al =1$. By  $\delta\le 1/2$ we therefore obtain
$$
 \P\big\{ \sum_{\al\in  {\cal I}}\xi_\al \leq \delta m\big\}
\leq \P\big\{ |X| \ge \frac{m}{2}\big\}\; .
$$
Since
$$
  A^2:= \sum_{i,j=1}^N |a_{ij}|^2 = \sum_{\al,\beta\in  {\cal I}}
\sum_{i,j=1}^N{v}_\al(i)v_\al(j)v_\beta(i){v}_\beta(j) = m \; ,
$$
by Proposition \ref{prop:HW}, we obtain
$$
 \P\big\{ \sum_{\al\in  {\cal I}}\xi_\al \leq \delta m\big\}
\leq \P\big\{ |X| \ge \frac{m}{2}\big\}\;
\leq 4\exp\Big( -c\min\big\{
\frac{m}{2A}, \frac{m^2}{4A^2}
\big\} \Big)\leq  e^{-c\sqrt{m}}.
$$
for some $c>0$. \qed

\bigskip
\noindent
{\bf Acknowledgement}: We thank  Jun Yin  for
several helpful comments and pointing out some errors
in the preliminary  versions of this paper. 
We are also grateful to the referees for their suggestions
to improve the presentation.


\bigskip
\thebibliography{hhhh}



\bibitem{BE} Bakry, D.,  \'Emery, M.: Diffusions hypercontractives. in: 
S\'eminaire
de probabilit\'es, XIX, 1983/84, {\bf 1123} Lecture Notes in Mathematics, 
Springer,
Berlin, 1985, 177--206.




\bibitem{BP} Ben Arous, G., P\'ech\'e, S.: Universality of local
eigenvalue statistics for some sample covariance matrices.
{\it Comm. Pure Appl. Math.} {\bf LVIII.} (2005), 1--42.


\bibitem{BI} Bleher, P.,  Its, A.: Semiclassical asymptotics of
orthogonal polynomials, Riemann-Hilbert problem, and universality
in the matrix model. {\it Ann. of Math.} {\bf 150} (1999): 185--266.

\bibitem{BG}
Bobkov, S. G., G\"otze, F.: Exponential integrability
and transportation cost related to logarithmic
Sobolev inequalities. {\it J. Funct. Anal.} {\bf 163} (1999), no. 1, 1--28.


\bibitem{BH} Br\'ezin, E., Hikami, S.: Correlations of nearby levels induced
by a random potential. {\it Nucl. Phys. B} {\bf 479} (1996), 697--706, and
Spectral form factor in a random matrix theory. {\it Phys. Rev. E}
{\bf 55} (1997), 4067--4083.


\bibitem{DKMVZ1} Deift, P., Kriecherbauer, T., McLaughlin, K.T-R,
Venakides, S., Zhou, X.: Uniform asymptotics for polynomials
orthogonal with respect to varying exponential weights and applications
to universality questions in random matrix theory.
{\it  Comm. Pure Appl. Math.} {\bf 52} (1999):1335--1425.

\bibitem{DKMVZ2} Deift, P., Kriecherbauer, T., McLaughlin, K.T-R,
Venakides, S., Zhou, X.: Strong asymptotics of orthogonal polynomials
with respect to exponential weights.
{\it  Comm. Pure Appl. Math.} {\bf 52} (1999): 1491--1552.


\bibitem{Dy} Dyson, F.J.: A Brownian-motion model for the eigenvalues
of a random matrix. {\it J. Math. Phys.} {\bf 3}, 1191-1198 (1962).



\bibitem{ESY1} Erd{\H o}s, L., Schlein, B., Yau, H.-T.:
Semicircle law on short scales and delocalization
of eigenvectors for Wigner random matrices.
{\it Ann. Probab.} {\bf 37}, No. 3, 815--852 (2009)

\bibitem{ESY3} Erd{\H o}s, L., Schlein, B., Yau, H.-T.:
Wegner estimate and level repulsion for Wigner random matrices.
{\it Int. Math. Res. Notices.} {\bf 2010}, No. 3, 436-479 (2010)

\bibitem{ERSY}  Erd{\H o}s, L., Ramirez, J., Schlein, B., Yau, H.-T.:
Universality of sine-kernel for Wigner matrices with a small Gaussian
 perturbation. {\it Electr. J. Prob.} {\bf 15},  Paper 18, 526--604 (2010)


\bibitem{EPRSY}
Erd\H{o}s, L., P\'ech\'e, S., Ram\'irez, J.,  Schlein, B. and Yau, H.-T.: Bulk 
universality
for Wigner matrices. {\it Commun. Pure  Applied Math.}
{\bf 63},  895-925, (2010) 
% Preprint arXiv.org:0905.4176.

\bibitem{ERSTVY}
Erd\H{o}s, L., Ram\'irez, J.,  Schlein, B., Tao, T., Vu, V. and Yau, H.-T.:
Bulk universality for Wigner hermitian matrices with subexponential decay.
To appear in Math. Res. Letters. Preprint arXiv:0906.4400


\bibitem{ESYY} Erd{\H o}s, L., Schlein, B., Yau, H.-T., Yin, J.:
The local relaxation flow approach to universality of the local
statistics for random matrices.  
% To  Annales Inst. H. Poincar\'e, Prob. and Stat. 
Preprint arXiv:0911.3687

\bibitem{EYY} Erd{\H o}s, L.,  Yau, H.-T., Yin, J.: 
Bulk universality for generalized Wigner matrices. 
%Submitted to Probab. Theor. Rel. Fields. 
Preprint arXiv:1001.3453



\bibitem{EYY2} Erd\H{o}s, L., Yau, H.-T., Yin, J.: 
Universality for generalized Wigner matrices with Bernoulli distribution.
Preprint arXiv:1003.3813

\bibitem{EYY3} Erd\H{o}s, L., Yau, H.-T., Yin, J.: 
Rigidity of Eigenvalues of Generalized Wigner Matrices.
Preprint arxiv:1007.4652



\bibitem{G} Guionnet, A.: Large random matrices: Lectures
on Macroscopic Asymptotics. \'Ecole d'E\'t\'e de Probabilit\'es
de Saint-Flour XXXVI-2006. Springer.

\bibitem{Gu} Gustavsson, J.: Gaussian fluctuations of eigenvalues in the
GUE, {\it Ann. Inst. H. Poincar\'e, Probab. Statist.} {\bf 41} (2005), no.2,
151--178


\bibitem{HW} Hanson, D.L., Wright, F.T.: A bound on
tail probabilities for quadratic forms in independent random
variables. {\it The Annals of Math. Stat.} {\bf 42} (1971), no.3,
1079-1083.



\bibitem{J} Johansson, K.: Universality of the local spacing
distribution in certain ensembles of Hermitian Wigner matrices.
{\it Comm. Math. Phys.} {\bf 215} (2001), no.3. 683--705.

\bibitem{Le} Ledoux, M.: The concentration of measure phenomenon.
 Mathematical Surveys and Monographs, {\bf 89}
      American Mathematical Society, Providence, RI, 2001.

\bibitem{LuY}
Lu, S.-L. and Yau, H.-T.: Spectral gap and logarithmic Sobolev
inequality for Kawasaki and Glauber dynamics,
{\em Comm. Math. Phys.} {\bf 156}, 399--433, 1993.


\bibitem{M} Mehta, M.L.: Random Matrices. Academic Press, New York, 1991.

\bibitem{PS} Pastur, L., Shcherbina M.:
Bulk universality and related properties of Hermitian matrix models.
J. Stat. Phys. {\bf 130} (2008), no.2., 205-250.




\bibitem{SS} Sinai, Y. and Soshnikov, A.:
A refinement of Wigner's semicircle law in a neighborhood of the spectrum edge.
{\it Functional Anal. and Appl.} {\bf 32} (1998), no. 2, 114--131.


\bibitem{Sosh} Soshnikov, A.: Universality at the edge of the spectrum in
Wigner random matrices. {\it  Comm. Math. Phys.} {\bf 207} (1999), no.3.
697-733.


\bibitem{TV} Tao, T. and Vu, V.: Random matrices: Universality of the
local eigenvalue statistics.
 Preprint arXiv:0906.0510.

\bibitem{Vu} Vu, V.: Spectral norm of random matrices. {\it Combinatorica},
{\bf 27} (6) (2007), 721-736.


\bibitem{Wr} Wright, F.T.: A bound on tail probabilities for quadratic
forms in independent random variables whose distributions are not
necessarily symmetric. {\it Ann. Probab.} {\bf 1} No. 6. (1973),
1068-1070.

\bibitem{Y} Yau, H. T.: Relative entropy and the hydrodynamics
of Ginzburg-Landau models, {\it Lett. Math. Phys}. {\bf 22} (1991) 63--80.
\end{document}